\documentclass[a4paper,UKenglish]{lipics-v2016}

\usepackage[]{graphicx}\usepackage[]{color}
\makeatletter
\def\maxwidth{ %
  \ifdim\Gin@nat@width>\linewidth
    \linewidth
  \else
    \Gin@nat@width
  \fi
}
\makeatother

\definecolor{fgcolor}{rgb}{0.345, 0.345, 0.345}

\usepackage{framed}
\makeatletter
 {\par\unskip\endMakeFramed%
 \at@end@of@kframe}
\makeatother

\definecolor{shadecolor}{rgb}{.97, .97, .97}
\definecolor{messagecolor}{rgb}{0, 0, 0}
\definecolor{warningcolor}{rgb}{1, 0, 1}
\definecolor{errorcolor}{rgb}{1, 0, 0}
\newenvironment{knitrout}{}{} 

\usepackage{alltt}
\newcommand{\SweaveOpts}[1]{}  
\newcommand{\SweaveInput}[1]{} 
\newcommand{\Sexpr}[1]{}       

\usepackage{numprint}

\usepackage{booktabs,array,multirow}
\newcolumntype{C}[1]{>{\centering}m{#1}}
\usepackage{etoolbox,siunitx}
\robustify\bfseries

 \usepackage{tikz}
 \usetikzlibrary{calc}
 \usepackage{pgfplots}
 \usepgfplotslibrary{external}
\usepackage{diagbox}
\usepackage{longtable}
\usepackage{pifont}

\usepackage{microtype}

\DeclareFontFamily{OT1}{pzc}{}
\DeclareFontShape{OT1}{pzc}{m}{it}{<-> s * [1.10] pzcmi7t}{}
\DeclareMathAlphabet{\mathpzc}{OT1}{pzc}{m}{it}

\usepackage[algo2e,ruled,vlined,linesnumbered]{algorithm2e}
\SetCommentSty{textit}
\DontPrintSemicolon
\IncMargin{-\parindent}
\SetAlCapHSkip{0pt}
\SetAlgoLined
\makeatletter
\newcommand{\removelatexerror}{\let\@latex@error\@gobble}
\makeatother

\SetKwProg{Function}{Function}{}{}

\newcommand{\Rem}[1]{\tcp*{#1}}
\newcommand{\Remi}[1]{\tcp*[f]{#1}}

\title{Network Flow-Based Refinement for Multilevel Hypergraph Partitioning}

\author[1]{Tobias Heuer}
\author[2]{Peter Sanders}
\author[3]{Sebastian Schlag}
\affil[1]{Karlsruhe Institute of Technology, Karlsruhe, Germany\\
  \texttt{tobias.heuer@gmx.net}}
\affil[2]{Karlsruhe Institute of Technology, Karlsruhe, Germany\\
  \texttt{sanders@kit.edu}}
\affil[3]{Karlsruhe Institute of Technology, Karlsruhe, Germany\\
  \texttt{sebastian.schlag@kit.edu}}
\authorrunning{T. Heuer, P. Sanders and S. Schlag}

\Copyright{Tobias Heuer, Peter Sanders and Sebastian Schlag}


\usepackage{amsfonts}

\newcommand{\realrange}[2]{\left[#1, #2\right]}

\newcommand{\unitrange}[2]{\realrange{0}{1}}

\newcommand{\Oh}[1]{\mathcal{O}\!\left( #1\right)}

\newcommand{\llabel}[1]{\label{\labelprefix:#1}}
\newcommand{\labelprefix}{} 

\newcommand{\discussionsize}{\small}

\newcommand{\frage}[1]{}

\newenvironment{code}{\noindent
\begin{tabbing}%
\hspace{2em}\=\hspace{2em}\=\hspace{2em}\=\hspace{2em}\=\hspace{2em}\=%
\hspace{2em}\=\hspace{2em}\=\hspace{2em}\=\hspace{2em}\=\hspace{2em}\=%
\kill}{\end{tabbing}}

\newcommand{\labelcommand}{}
\newcommand{\captiontext}{}
\newsavebox{\codeparam}
\newcounter{lineNumber}
\newenvironment{disscodepos}[3]{%
\renewcommand{\labelcommand}{#2}%
\renewcommand{\captiontext}{#3}%
\sbox{\codeparam}{\parbox{\textwidth}{#3}}%
\begin{figure}[#1]\begin{center}\begin{code}\setcounter{lineNumber}{1}}{%
\end{code}\end{center}\caption{\llabel{\labelcommand}\captiontext}\end{figure}}

{\end{disscodepos}}

\newdimen\endofsize\endofsize=0.5em
\def\endofbeweis{~\quad\hglue\hsize minus\hsize
                 \hbox{\vrule height \endofsize width
\endofsize}\par}

\newcommand{\ALL}{All}
\newcommand{\Dual}{Dual}
\newcommand{\Primal}{Primal}
\newcommand{\ISPD}{ISPD98}
\newcommand{\Literal}{Literal}
\newcommand{\SPM}{SPM}
\newcommand{\DAC}{DAC}

\newcommand{\KaHyPar}[1]{KaHyPar-#1}

\newcommand{\hMetis}[1]{hMetis-#1}
\newcommand{\PaToH}[1]{PaToH-#1}

\newcommand{\FlowVariant}[3]{(#1\text{F},#2\text{M},#3\text{FM})}
\newcommand{\Constant}[1]{\textsc{Constant#1}}

\begin{document}

\maketitle

\begin{abstract}
  We present a refinement framework for multilevel hypergraph partitioning that uses max-flow computations on pairs of blocks to improve
  the solution quality of a $k$-way partition. The framework generalizes the flow-based improvement algorithm of KaFFPa
  from graphs to hypergraphs and is integrated into the hypergraph partitioner KaHyPar.
  By reducing the size of hypergraph flow networks, improving the flow model used in KaFFPa, and
  developing techniques to improve the running time of our algorithm, we obtain a partitioner that computes the best solutions
  for a wide range of benchmark hypergraphs from different application areas while still having a running time comparable to that of hMetis.
\end{abstract}

\section{Introduction}
Given an undirected hypergraph $H=(V,E)$, the \emph{$k$-way hypergraph partitioning problem} 
is to partition the vertex set into $k$ disjoint blocks of bounded size (at most $1+\varepsilon$ times the average block size)
such that an objective function involving the cut hyperedges is minimized.
Hypergraph partitioning (HGP) has many important applications in practice such as  
scientific computing~\cite{PaToH} or VLSI design~\cite{Papa2007}.
Particularly VLSI design is a field where small improvements can lead to significant savings~\cite{Wichlund:1998}.

It is well known that HGP is NP-hard~\cite{Lengauer:1990}, which is why practical applications mostly use heuristic \emph{multilevel} algorithms~\cite{MultiLevel_Bui,MultiLevel_Cong,MultiLevel_Hauck,MultiLevel_Hendrickson}. 
These algorithms successively \emph{contract} the hypergraph to obtain a hierarchy of smaller, structurally similar hypergraphs.
After applying an \emph{initial partitioning} algorithm to the smallest hypergraph, contraction is undone and, at each level, a
\emph{local search} method is used to improve the partitioning induced by the coarser level.
\emph{All} state-of-the-art HGP algorithms~\cite{ahss2017alenex,MLPart,Aykanat:2008,Zoltan,hs2017sea,SHP,hMetisRB,hMetisKway,KaHyPar-R,Parkway2.0,DBLP:conf/dimacs/CatalyurekDKU12,Mondriaan}
either use variations of the Kernighan-Lin (KL)~\cite{KLAlgorithm,Schweikert:1972} or the Fiduccia-Mattheyses (FM) heuristic~\cite{FM82,HypergraphKFM}, or
simpler greedy algorithms~\cite{hMetisRB,hMetisKway} for local search.
These heuristics move vertices between blocks in descending order of improvements in the optimization objective (gain) and
are known to be prone to get stuck in local optima when used directly on the input hypergraph~\cite{hMetisKway}.
The multilevel paradigm helps to some extent, since 
it allows a more global view on the problem on the coarse levels and a very fine-grained view on the fine levels of the multilevel hierarchy.
However, the performance of move-based approaches degrades for hypergraphs with large hyperedges.
In these cases, it is difficult to find meaningful vertex moves that improve the solution quality because large hyperedges are
likely to have many vertices in multiple blocks~\cite{doi:10.1137/S1064827502410463}. Thus the gain of moving a single vertex to another block is likely to be \emph{zero}~\cite{FormulaPartitioning14}.

While finding \emph{balanced} minimum cuts in hypergraphs is NP-hard, a minimum cut separating two vertices can be found
in polynomial time using network flow algorithms and the well-known max-flow min-cut theorem~\cite{ford1956maximal}. Flow algorithms find an optimal min-cut
and do not suffer the drawbacks of move-based approaches. However, they were long overlooked as heuristics
for balanced partitioning due to their high complexity~\cite{GraphEdgeReduction,yang1996balanced}.
In the context of graph partitioning, Sanders and Schulz~\cite{kaffpa} recently 
presented a max-flow-based improvement algorithm which is integrated into the multilevel partitioner KaFFPa and computes high quality solutions.

\subparagraph*{Outline and Contribution.}\label{OutlineContribution}
Motivated by the results of Sanders and Schulz~\cite{kaffpa}, we generalize the max-flow min-cut refinement framework
of KaFFPa from graphs to hypergraphs.
After introducing basic notation and giving a brief overview of related work and the
techniques used in KaFFPa in Section~\ref{Preliminaries}, we explain how hypergraphs are transformed into
flow networks and present a technique to reduce the size of the resulting hypergraph flow network in Section~\ref{HypergraphFlows}.
In Section~\ref{SourcesSinks} we then show how this network can be used to construct a flow problem such that the
min-cut induced by a max-flow computation between a pair of blocks improves the solution quality of a $k$-way partition.
We furthermore identify shortcomings of the KaFFPa approach that restrict the search space of feasible solutions significantly and
introduce an advanced model that overcomes these limitations by exploiting the structure of hypergraph flow networks.
We implemented our algorithm in the open source HGP framework KaHyPar and therefore briefly discuss implementation details and techniques to improve
the running time in Section~\ref{KaHyParIntegration}.
Extensive experiments presented in Section~\ref{Experiments} demonstrate that our flow model yields better solutions than the KaFFPa approach
for both hypergraphs \emph{and} graphs. We furthermore show that using pairwise flow-based refinement significantly improves
partitioning quality. The resulting hypergraph partitioner, \KaHyPar{MF},  performs better than \emph{all} competing algorithms on \emph{all} instance classes
and still has a running time comparable to that of hMetis.
On a large benchmark set consisting of 3222 instances from various application domains, \KaHyPar{MF} computes the best partitions in 2427 cases.

\section{Preliminaries}\label{Preliminaries}
\subsection{Notation and Definitions}\label{Notations}
An \textit{undirected hypergraph} $H=(V,E,c,\omega)$ is defined as a set of $n$ vertices $V$ and a
set of $m$ hyperedges/nets $E$ with vertex weights $c:V \rightarrow \mathbb{R}_{>0}$ and net 
weights $\omega:E \rightarrow \mathbb{R}_{>0}$, where each net is a subset of the vertex set $V$ (i.e., $e \subseteq V$). The vertices of a net are called \emph{pins}. 
We extend $c$ and $\omega$ to sets, i.e., $c(U) :=\sum_{v\in U} c(v)$ and $\omega(F) :=\sum_{e \in F} \omega(e)$.
A vertex $v$ is \textit{incident} to a net $e$ if $v \in e$. $\mathrm{I}(v)$ denotes the set of all incident nets of $v$. 
The \textit{degree} of a vertex $v$ is $d(v) := |\mathrm{I}(v)|$. 
The \textit{size} $|e|$ of a net $e$ is the number of its pins.
Given a subset $V' \subset V$, the \emph{subhypergraph} $H_{V'}$ is defined as $H_{V'}:=(V', \{e \cap V'~|~e \in E : e \cap V' \neq \emptyset \})$.

A \emph{$k$-way partition} of a hypergraph $H$ is a partition of its vertex set into $k$ \emph{blocks} $\mathrm{\Pi} = \{V_1, \dots, V_k\}$ 
such that $\bigcup_{i=1}^k V_i = V$, $V_i \neq \emptyset $ for $1 \leq i \leq k$, and $V_i \cap V_j = \emptyset$ for $i \neq j$.
We call a $k$-way partition $\mathrm{\Pi}$ \emph{$\mathrm{\varepsilon}$-balanced} if each block $V_i \in \mathrm{\Pi}$ satisfies the \emph{balance constraint}:
$c(V_i) \leq L_{\max} := (1+\varepsilon)\lceil \frac{c(V)}{k} \rceil$ for some parameter $\mathrm{\varepsilon}$. 
For each net $e$, $\mathrm{\Lambda}(e) := \{V_i~|~ V_i \cap e \neq \emptyset\}$ denotes the \emph{connectivity set} of $e$.
The \emph{connectivity} of a net $e$ is $\mathrm{\lambda}(e) := |\mathrm{\Lambda}(e)|$.
A net is called \emph{cut net} if $\mathrm{\lambda}(e) > 1$.
Given a $k$-way partition $\mathrm{\Pi}$ of $H$, the \emph{quotient graph} $Q := (\mathrm{\Pi}, \{(V_i,V_j)~|~\exists e \in E : \{V_i,V_j\} \subseteq  \mathrm{\Lambda}(e)\})$ contains an edge between each pair of adjacent blocks.
The \emph{$k$-way hypergraph partitioning problem} is to find an $\varepsilon$-balanced $k$-way partition $\mathrm{\Pi}$ of a hypergraph $H$ that
minimizes an objective function over the cut nets for some~$\varepsilon$.
Several objective functions exist in the literature~\cite{Alpert19951,Lengauer:1990}.
The most commonly used cost functions are the \emph{cut-net} metric $\text{cut}(\mathrm{\Pi}) := \sum_{e \in E'} \omega(e)$ and the
\emph{connectivity} metric $(\mathrm{\lambda} - 1)(\mathrm{\Pi}) := \sum_{e\in E'} (\mathrm{\lambda}(e) -1)~\omega(e)$~\cite{ConnecivityMetric}, where $E'$ is the set of all cut nets~\cite{donath1988logic}.
In this paper, we use the $(\lambda - 1)$-metric.
Optimizing both objective functions is known to be NP-hard \cite{Lengauer:1990}.
Hypergraphs can be represented as \emph{bipartite} graphs~\cite{HuMoerder85}.
In the following, we use \emph{nodes} and \emph{edges} when referring to graphs and \emph{vertices}
and \emph{nets} when referring to hypergraphs. In the bipartite graph $G_*(V \dot\cup E, F)$ the vertices and nets of $H$ form the node set and
for each net $e \in \mathrm{I}(v)$, we add an edge $(e,v)$ to $G_*$. The edge set $F$ is thus defined as $F := \{(e,v)~|~e \in E, v \in e \}$. Each net in $E$ therefore corresponds to a \emph{star} in $G_*$.

Let $G=(V,E,c,\omega)$ be a weighted directed graph. We use the same notation as for hypergraphs to refer to node weights $c$,
edge weights $\omega$, and node degrees $d(v)$. Furthermore $\mathrm{\Gamma(u)} := \{v : (u,v) \in E\}$ denotes the neighbors of node $u$.
A \emph{path} $P = \langle v_1,\ldots,v_k\rangle$ is a sequence of nodes, such that each pair of consecutive nodes is connected by an edge.
A \emph{strongly connected component} $C \subseteq V$ is a set of nodes such that for each $u,v \in C$
there exists a path from $u$ to $v$. A \emph{topological} ordering is a linear ordering $\prec$ of $V$ such that every
directed edge $(u,v) \in E$ implies $u \prec v$ in the ordering. A set of nodes $B \subseteq V$ is called a \emph{closed set} iff
there are no outgoing edges leaving $B$, i.e., if the conditions $u \in B$ and $(u,v) \in E$ imply $v \in B$.
A subset $S\subset V$ is called a \emph{node separator} if its removal divides $G$ into two disconnected components.

A flow network $\mathcal{N}=(\mathcal{V},\mathcal{E},\mathpzc{c})$ is a directed graph with two distinguished nodes $\mathpzc{s}$ and $\mathpzc{t}$
in which each edge $e \in \mathcal{E}$ has a capacity $\mathpzc{c}(e) \geq 0$. An $(\mathpzc{s},\mathpzc{t})$-flow (or flow)
is a function $f: \mathcal{V} \times \mathcal{V} \rightarrow \mathbb{R}$ that satisfies the \emph{capacity constraint} $\forall u,v \in \mathcal{V}: f(u,v) \le \mathpzc{c}(u,v)$,
the \emph{skew symmetry constraint} $\forall v \in \mathcal{V} \times \mathcal{V}: f(u,v) = -f(v,u)$, and
the \emph{flow conservation constraint} $\forall u \in \mathcal{V} \setminus \{\mathpzc{s},\mathpzc{t}\}: \sum_{v \in \mathcal{V}} f(u,v) = 0$. The value of a flow $|f| := \sum_{v \in \mathcal{V}}{f(\mathpzc{s},v)}$ is defined as the total amount
of flow transferred from $\mathpzc{s}$ to $\mathpzc{t}$.
The \emph{residual capacity} is defined as $r_f(u,v) = \mathpzc{c}(u,v) - f(u,v)$.
Given a flow $f$, $\mathcal{N}_f = (\mathcal{V},\mathcal{E}_f, r_f)$ with $\mathcal{E}_f= \{(u,v) \in \mathcal{V} \times \mathcal{V}\ |\ r_f(u,v) > 0\}$
is the \emph{residual network}. An  $(\mathpzc{s},\mathpzc{t})$-cut (or cut) is a bipartition $(\mathcal{S}, \mathcal{V}\setminus \mathcal{S})$ of a flow network $\mathcal{N}$  with $\mathpzc{s} \in \mathcal{S} \subset \mathcal{V}$ and $\mathpzc{t} \in \mathcal{V} \setminus \mathcal{S}$. The capacity of an $(\mathpzc{s},\mathpzc{t})$-cut is defined as $\sum_{e \in \mathcal{E}'}\mathpzc{c}(e)$, where $\mathcal{E}'= \{(u,v) \in \mathcal{E}: u \in \mathcal{S}, v \in \mathcal{V} \setminus \mathcal{S}\}$.
The max-flow min-cut theorem states that the value $|f|$ of a maximum flow is equal to the capacity of a minimum cut
separating $\mathpzc{s}$ and $\mathpzc{t}$~\cite{ford1956maximal}. 

\subsection{Related Work}\label{RelatedWork}
\subparagraph*{Hypergraph Partitioning.}\label{Overview}
Driven by applications in VLSI design and scientific computing, HGP has evolved into a broad research area since the 1990s.
We refer to \cite{Alpert19951,DBLP:conf/dimacs/2012,Papa2007,trifunovic2006parallel} for an extensive overview. 
Well-known multilevel HGP software packages with certain distinguishing characteristics include PaToH~\cite{PaToH} (originating from scientific computing),
hMetis~\cite{hMetisRB,hMetisKway} (originating from VLSI design), KaHyPar~\cite{ahss2017alenex,hs2017sea,KaHyPar-R} (general purpose, $n$-level), Mondriaan~\cite{Mondriaan} (sparse matrix
partitioning), MLPart~\cite{MLPart} (circuit partitioning), Zoltan~\cite{Zoltan}, Parkway~\cite{Parkway2.0}  and SHP~\cite{SHP}  (distributed),
UMPa~\cite{DBLP:conf/dimacs/CatalyurekDKU12} (directed hypergraph model, multi-objective), and kPaToH (multiple constraints, fixed vertices)~\cite{Aykanat:2008}.
\emph{All} of these tools either use variations of the Kernighan-Lin (KL)~\cite{KLAlgorithm,Schweikert:1972} or the Fiduccia-Mattheyses (FM) heuristic~\cite{FM82,HypergraphKFM},
or algorithms that greedily move vertices~\cite{hMetisKway} or nets~\cite{hMetisRB} to improve solution quality in the refinement phase.

\subparagraph*{Flows on Hypergraphs.}\label{HypergraphFlows}
While flow-based approaches have not yet been considered as refinement algorithms for multilevel HGP, several works deal with flow-based hypergraph min-cut computation.
The problem of finding minimum $(\mathpzc{s},\mathpzc{t})$-cuts in hypergraphs was first considered by Lawler~\cite{lawler1973}, who showed that it can be reduced to computing
maximum flows in directed graphs.
Hu and Moerder~\cite{HuMoerder85} present an augmenting path algorithm to compute a minimum-weight vertex separator on the
star-expansion of the hypergraph. Their vertex-capacitated network can also be transformed into
an edge-capacitated network using a transformation due to Lawler~\cite{lawler2001combinatorial}.
Yang and Wong~\cite{yang1996balanced} use repeated, incremental max-flow min-cut computations on the Lawler network~\cite{lawler1973} to find $\varepsilon$-balanced hypergraph bipartitions.
Solution quality and running time of this algorithm are improved by Lillis and Cheng~\cite{480016} by introducing advanced heuristics to select source and sink nodes.
Furthermore, they present a preflow-based~\cite{GoldbergT88} min-cut algorithm that implicitly operates on the star-expanded hypergraph.
Pistorius and Minoux~\cite{pistorius2003} generalize the algorithm of Edmonds and Karp~\cite{edmonds1972theoretical} to hypergraphs by labeling both vertices and nets.
Liu and Wong~\cite{GraphEdgeReduction} simplify Lawler's  hypergraph flow network~\cite{lawler1973} by explicitly distinguishing between graph edges
and hyperedges with three or more pins.  This approach significantly reduces the size of flow networks derived from VLSI hypergraphs, since most of the nets in a circuit are graph edges.
Note that the above-mentioned approaches to model hypergraphs as flow networks for max-flow min-cut computations do not contradict the negative results of Ihler et al.~\cite{IhlerWW93}, who show
that, \emph{in general}, there does not exist an edge-weighted graph $G=(V,E)$ that correctly represents the min-cut properties of the corresponding hypergraph $H=(V,E)$.

\subparagraph*{Flow-Based Graph Partitioning.}\label{GraphFlows}
Flow-based refinement algorithms for graph partitioning include Improve~\cite{andersen2008algorithm} and MQI~\cite{lang2004flow}, which improve expansion or conductance of
bipartitions. MQI also
yields as small improvement when used as a post processing technique on hypergraph bipartitions initially computed by hMetis~\cite{lang2004flow}.
 FlowCutter~\cite{HamannS16} uses an approach similar to Yang and Wong~\cite{yang1996balanced} to compute graph bisections that are Pareto-optimal in regard to cut size and balance.
 Sanders and Schulz~\cite{kaffpa} present a flow-based refinement framework for their direct $k$-way  graph partitioner KaFFPa. The algorithm works on pairs of adjacent blocks and constructs flow problems such that each min-cut in the flow network is a feasible solution in regard to the original partitioning problem. 

\subparagraph*{KaHyPar.}\label{KaHyPar}
Since our algorithm is integrated into the KaHyPar framework, we briefly review its core components.
While traditional multilevel HGP algorithms contract matchings or clusterings and therefore
work with a coarsening hierarchy of $\Oh{\log n}$ levels, KaHyPar instantiates the multilevel paradigm in the
extreme $n$-level version, removing only a \emph{single} vertex between two levels.
After coarsening, a portfolio of simple algorithms is used to create an initial partition of the coarsest hypergraph. During uncoarsening,
strong localized local search heuristics based on the FM algorithm~\cite{FM82,HypergraphKFM} are used to refine
the solution. 
Our work builds on KaHyPar-CA~\cite{hs2017sea}, which is a direct $k$-way partitioning algorithm for optimizing the $(\lambda-1)$-metric. It uses an
improved coarsening scheme that incorporates global information about the community structure of the hypergraph into the coarsening process.


\subsection{The Flow-Based Improvement Framework of KaFFPa}\label{KaFFPa}
We discuss the framework of Sanders and Schulz~\cite{kaffpa} in greater detail, since our work makes use of the techniques proposed by the authors.
For simplicity, we assume $k=2$. The techniques can be applied on a $k$-way partition by repeatedly executing the algorithm on pairs of adjacent blocks.
To schedule these refinements, the authors propose an \emph{active block scheduling} algorithm, which schedules blocks as long as
their participation in a pairwise refinement step results in some changes in the $k$-way partition.

An $\varepsilon$-balanced bipartition of a graph $G = (V,E,c,\omega)$ is improved with flow computations as follows.
The basic idea is to construct a flow network $\mathcal{N}$ based on the induced subgraph $G[B]$, where $B \subseteq V$ is a set of nodes around
the cut of $G$. The size of $B$ is controlled by an imbalance factor $\varepsilon' := \alpha \varepsilon$, where $\alpha$ is a scaling parameter that is chosen adaptively depending on the result of the min-cut computation.
If the heuristic found an $\varepsilon$-balanced partition using $\varepsilon'$, the cut is accepted and $\alpha$ is increased to $\min(2\alpha, \alpha')$ where $\alpha'$ is a predefined upper bound.
Otherwise it is decreased to $\max(\frac{\alpha}{2},1)$. This scheme continues until a maximal number of rounds is reached or a feasible
partition that did not improve the cut is found.

In each round, the corridor $B := B_1 \cup B_2$ is constructed by performing two breadth-first searches (BFS). The first BFS is done in the induced subgraph $G[V_1]$. It is
initialized with the boundary nodes of $V_1$ and stops if $c(B_1)$ would exceed $(1+\epsilon') \lceil \frac{c(V)}{2} \rceil - c(V_2)$. The second BFS constructs $B_2$ in an analogous fashion
using $G[V_2]$.
Let $\delta B := \{u \in B\ |\ \exists (u,v) \in E: v \notin B\}$ be the border of $B$. Then $\mathcal{N}$ is constructed by
connecting all border nodes  $\delta B \cap V_1$ of $G[B]$ to the source $\mathpzc{s}$ and all border nodes $\delta B \cap V_2$ to the sink $\mathpzc{t}$ using directed edges with
an edge weight of $\infty$. By connecting $\mathpzc{s}$ and $\mathpzc{t}$ to the respective border nodes, it is ensured that  edges incident to border nodes, but not contained in $G[B]$, cannot become
cut edges.
For $\alpha = 1$, the size of $B$ thus ensures that the flow network $\mathcal{N}$ has the \emph{cut property}, i.e., each $(\mathpzc{s},\mathpzc{t})$-min-cut in $\mathcal{N}$ yields an $\varepsilon$-balanced partition of $G$ with a possibly smaller cut. For larger values of $\alpha$, this does not have to be the case.

After computing a max-flow in $\mathcal{N}$, the algorithm tries to find a min-cut with better balance.
This is done by exploiting the fact that \emph{one} $(\mathpzc{s},\mathpzc{t})$-max-flow contains information about \emph{all} $(\mathpzc{s},\mathpzc{t})$-min-cuts~\cite{picard1980structure}.
More precisely, the algorithm uses the 1--1 correspondence between $(\mathpzc{s},\mathpzc{t})$-min-cuts and closed sets containing $\mathpzc{s}$ in the Picard-Queyranne-DAG $D_{\mathpzc{s},\mathpzc{t}}$ of the residual graph $\mathcal{N}_f$~\cite{picard1980structure}.
First, $D_{\mathpzc{s},\mathpzc{t}}$ is constructed by contracting each strongly connected component of the residual graph.
Then the following heuristic (called most balanced minimum cuts) is repeated several times using different random seeds.
Closed node sets containing $s$ are computed by sweeping through the nodes of $DAG_{\mathpzc{s},\mathpzc{t}}$  in reverse topological order (e.g. computed using a randomized DFS).
Each closed set induces a differently balanced min-cut and the one with the best balance (with respect to the original balance constraint) is used as resulting bipartition.

\section{Hypergraph Max-Flow Min-Cut Refinement}
In the following, we generalize the flow-based refinement algorithm of KaFFPa to hypergraph partitioning.
In Section~\ref{HypergraphFlows} we first show how hypergraph flow networks $\mathcal{N}$ are constructed in general and introduce a technique
to reduce their size by removing \emph{low-degree} hypernodes.
Given a $k$-way partition $\mathrm{\Pi}_k=\{V_1,\dots,V_k\}$ of a hypergraph $H=(V,E)$, a pair of blocks $(V_i,V_j)$ adjacent in the quotient graph $Q$, and a corridor $B$,
Section~\ref{SourcesSinks} then explains how $\mathcal{N}$ is used to build a flow problem $\mathcal{F}$ based on a $B$-induced subhypergraph $H_B=(V_B,E_B)$.
The flow problem $\mathcal{F}$ is constructed such that an $(\mathpzc{s},\mathpzc{t})$-max-flow computation optimizes the \emph{cut} metric of the bipartition $\mathrm{\Pi}_2=(V_i,V_j)$ of $H_B$
and thus improves the $(\lambda -1)$-metric in $H$.
Section~\ref{KaHyParIntegration} then discusses the integration into KaHyPar and introduces several techniques
to speed up flow-based refinement. Algorithm~\ref{alg:HGPflows} gives a pseudocode description of the entire flow-based refinement framework.

\begin{algorithm2e}[t!]
\caption{Flow-Based Refinement}\label{alg:HGPflows}\normalsize
\LinesNumberedHidden
\SetKwFunction{refine}{$\FuncSty{MaxFlowMinCutRefinement}$}\SetKwFunction{proc}{proc}
\SetKwProg{myalg}{Algorithm}{}{}
\KwIn{Hypergraph $H$, $k$-way partition $\mathrm{\Pi}_k=\{V_1,\dots,V_k\}$, imbalance parameter $\varepsilon$.}
\SetKwRepeat{Do}{do}{while}
\myalg{\refine{$H, \mathrm{\Pi}_k$}}{
  $Q:=  \FuncSty{QuotientGraph}(H,\mathrm{\Pi}_k)$\\
  \While(\Remi{in the beginning all blocks are active}){$\exists$ active blocks $\in Q$} {
    \ForEach(\Remi{choose a pair of blocks}){$\{(V_i, V_j) \in Q~|~V_i \vee V_j~\text{is active}\}$ }{
      $\mathrm{\Pi_\text{old}} = \mathrm{\Pi}_\text{best} := \{V_i, V_j\} \subseteq \Pi_k$ \Rem{extract bipartition to be improved}
      $\varepsilon_\text{old} = \varepsilon_\text{best} := \FuncSty{imbalance}(\Pi_k)$ \Rem{imbalance of current $k$-way partition}
      $\alpha := \alpha'$ \Rem{use large $B$-corridor for first iteration}
      \Do(\Remi{adaptive flow iterations}){$\alpha \ge 1$} {
        $B := \FuncSty{computeB-Corridor}(H, \mathrm{\Pi}_\text{best}, \alpha\varepsilon)$ \Rem{as described in Section~\ref{KaFFPa}}
      $H_B := \FuncSty{SubHypergraph}(H, B)$ \Rem{create $B$-induced subhypergraph}
      $\mathcal{N}_B := \FuncSty{FlowNetwork}(H_B)$ \Rem{as described in Section~\ref{HypergraphFlows}}
        $\mathcal{F} := \FuncSty{FlowProblem}(\mathcal{N_B})$  \Rem{as described in Section~\ref{SourcesSinks}}
      $f := \FuncSty{maxFlow}(\mathcal{F})$ \Rem{compute maximum flow on $\mathcal{F}$}
        $\mathrm{\Pi}_f := \FuncSty{mostBalancedMinCut}(f, \mathcal{F})$  \Rem{as in Section~\ref{KaFFPa} \& \ref{HypergraphFlows}}
        $\varepsilon_f := \FuncSty{imbalance}(\mathrm{\Pi}_f \cup \mathrm{\Pi}_k \setminus \mathrm{\Pi}_\text{old})$ \Rem{imbalance of new $k$-way partition}
        \If(\Remi{found improvement}){$(\text{cut}(\mathrm{\Pi}_f) < \text{cut}(\mathrm{\Pi}_\text{best}) \wedge \varepsilon_f \le \varepsilon) \vee \varepsilon_f < \varepsilon_\text{best}$}{
          $\alpha := \min(2\alpha, \alpha'),~ \mathrm{\Pi}_\text{best} := \mathrm{\Pi}_f,~\varepsilon_\text{best} := \varepsilon_f$ \Remi{update $\alpha$, $\mathrm{\Pi}_\text{best}$,$\varepsilon_\text{best}$}
        } \lElse{
          $\alpha := \frac{\alpha}{2}$ \Remi{decrease size of $B$-corridor in next iteration}
        }
      }

     \If(\Remi{improvement found}){$\mathrm{\Pi}_\text{best} \neq \mathrm{\Pi}_\text{old}$}{
       $\mathrm{\Pi}_k := \mathrm{\Pi}_\text{best} \cup \mathrm{\Pi}_k \setminus \mathrm{\Pi}_\text{old}$ \Rem{replace $\mathrm{\Pi}_\text{old}$ with $\mathrm{\Pi}_\text{best}$}
       $\FuncSty{activateForNextRound}(V_i, V_j)$ \Rem{reactivate blocks for next round}
      }
    }
  }
  \Return $\mathrm{\Pi}_k$
}{}
\KwOut{improved $\varepsilon$-balanced $k$-way partition $\mathrm{\Pi}_k=\{V_1, \dots, V_k\}$}
\end{algorithm2e}

\subsection{Hypergraph Flow Networks}\label{HypergraphFlows}

\subparagraph*{The Liu-Wong Network~\cite{GraphEdgeReduction}.}
Given a hypergraph $H=(V,E,c,\omega)$ and two distinct nodes $\mathpzc{s}$ and $\mathpzc{t}$, an $(\mathpzc{s},\mathpzc{t})$-min-cut can be computed by finding a minimum-capacity cut in the following flow
network $\mathcal{N}=(\mathcal{V}, \mathcal{E})$:
\begin{itemize}
\item $\mathcal{V}$ contains all vertices in $V$.
\item For each multi-pin net $e \in E$ with $|e| \geq 3$, add two \emph{bridging} nodes $e'$ and $e''$ to $\mathcal{V}$ and a \emph{bridging} edge $(e',e'')$ with capacity $\mathpzc{c}(e',e'')=\omega(e)$ to $\mathcal{E}$. For each pin $p \in e$, add two edges $(p,e')$ and $(e'',p)$ with capacity $\infty$ to $\mathcal{E}$.
\item For each two-pin net  $e=(u,v) \in E$, add two \emph{bridging edges} $(u,v)$ and $(v,u)$ with capacity $\omega(e)$ to $\mathcal{E}$.
\end{itemize}
The flow network of Lawler~\cite{lawler1973} does not distinguish between two-pin and multi-pin nets. This increases the size of the network by two vertices and three edges per two-pin net.
Figure~\ref{fig:flow_networks} shows an example of the Lawler and Liu-Wong hypergraph flow networks as well as of our network described in the following paragraph.

\subparagraph*{Removing Low Degree Hypernodes.}
We further decrease the size of the network by using the observation that the problem of finding an $(\mathpzc{s},\mathpzc{t})$-min-cut of $H$ can be reduced
to finding a minimum-weight $(\mathpzc{s},\mathpzc{t})$-vertex-separator in the star-expansion, where the capacity of each star-node is the weight of the corresponding net and all other nodes (corresponding to vertices in $H$)
have \emph{infinite} capacity~\cite{HuMoerder85}.
Since the separator has to be a subset of the star-nodes,
it is possible to replace \emph{any} infinite-capacity node by adding a clique between all adjacent star-nodes without affecting the separator.
The key observation now is that an infinite-capacity node $v$ with degree $d(v)$ induces $2d(v)$ infinite-capacity edges in the Lawler network~\cite{lawler1973},
while a clique between star-nodes induces $d(v)(d(v)-1)$ edges. For hypernodes with $d(v) \leq 3$, it therefore holds that  $d(v)(d(v)-1)\le 2d(v)$.
Thus we can reduce the number of nodes and edges of the Liu-Wong network as follows. Before applying the transformation on the star-expansion of $H$,
we remove all infinite-capacity nodes $v$ corresponding to hypernodes with $d(v) \leq 3$ that are \emph{not} incident to any two-pin nets and add a clique between all star-nodes adjacent to $v$.
In case $v$ was a source or sink node, we create a multi-source multi-sink problem by adding all adjacent star-nodes to the set of sources resp. sinks~\cite{FordFulkerson}.

\begin{figure}[t!] 
  \centering
  \includegraphics[width=\textwidth]{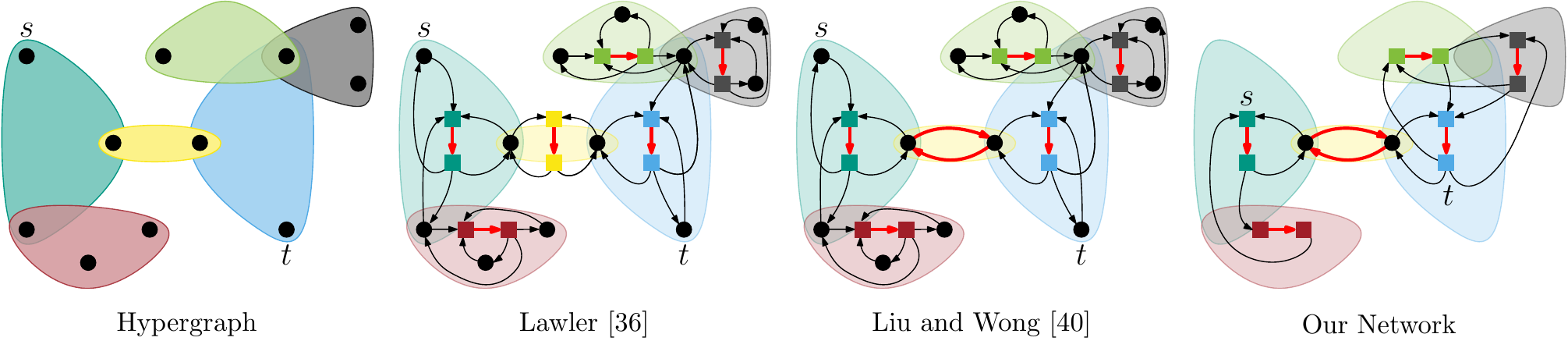}
  \caption{Illustration of hypergraph flow networks. Our approach further sparsifies the flow network of Liu and Wong~\cite{GraphEdgeReduction}. Thin edges have infinite capacity.}\label{fig:flow_networks}
\end{figure}

\subparagraph*{Reconstructing Min-Cuts.}
After computing an $(\mathpzc{s},\mathpzc{t})$-max-flow in the Lawler or Liu-Wong network, an $(\mathpzc{s},\mathpzc{t})$-min-cut of $H$ can be computed by a BFS in the residual graph starting from $\mathpzc{s}$.
Let $S$ be the set of nodes corresponding to vertices of $H$ reached by the BFS. Then $(S,V\setminus S)$ is an $(\mathpzc{s},\mathpzc{t})$-min-cut.
Since our network does not contain low degree hypernodes, we use the following lemma to compute an $(\mathpzc{s},\mathpzc{t})$-min-cut of $H$:

\begin{lemma}
\label{lemma:bipartition}
Let $f$ be a maximum $(\mathpzc{s},\mathpzc{t})$-flow in the Lawler network $\mathcal{N}=(\mathcal{V}, \mathcal{E})$ of a hypergraph $H=(V,E)$ and $(\mathcal{S},\mathcal{V}\setminus \mathcal{S})$ be the corresponding $(\mathpzc{s},\mathpzc{t})$-min-cut in $\mathcal{N}$.
Then for each node $v \in \mathcal{S} \cap V$, the residual graph $\mathcal{N}_f=(\mathcal{V}_f, \mathcal{E}_f)$ contains at least one path $\langle \mathpzc{s}, \dots, e'' \rangle$
to a bridging node $e''$ of a net $e \in \mathrm{I}(v)$.
\end{lemma}

\begin{proof}
  Since $v \in \mathcal{S}$, there has to be some path $\mathpzc{s} \rightsquigarrow v$ in $\mathcal{N}_f$. By definition of the flow network, this path can
  either be of the form $P_1=\langle \mathpzc{s}, \dots, e'', v\rangle$ or  $P_2=\langle \mathpzc{s}, \dots, e', v\rangle$ for some bridging nodes $e',e''$ corresponding to nets $e \in \mathrm{I}(v)$.
  In the former case we are done, since $e'' \in P_1$. In the latter case the existence of edge $(e',v) \in \mathcal{E}_f$ implies that
  there is a positive flow $f(v,e') > 0$ over edge $(v,e') \in \mathcal{E}$. Due to flow conservation, there exists at least one edge $(e'', v) \in \mathcal{E}$ with $f(e'',v)>0$,
  which implies that $(v, e'') \in \mathcal{E}_f$. Thus we can extend the path $P_2$ to $\langle \mathpzc{s}, \dots, e', v, e''\rangle$.
\end{proof}

Thus $(A, V \setminus A)$ is an $(\mathpzc{s},\mathpzc{t})$-min-cut of $H$, where $A := \{v \in e~|~\exists e \in E : \langle \mathpzc{s},\dots,e'' \rangle ~\text{in}~\mathcal{N}_f \}$.
Furthermore this allows us to search for more balanced
min-cuts using the Picard-Queyranne-DAG of $\mathcal{N}_f$ as described in Section~\ref{KaFFPa}.
By the definition of closed sets it follows that if a bridging node $e''$ is contained in a closed set $C$, then
all nodes $v \in \mathrm{\Gamma}(e'')$ (which correspond to vertices of $H$) are also contained in $C$. Thus we can use the respective bridging nodes $e''$ as representatives of removed low degree hypernodes.

\subsection{Constructing the Hypergraph Flow Problem}\label{SourcesSinks}
Let $H_B=(V_B, E_B)$ be the subhypergraph of $H=(V,E)$ that is induced by a corridor $B$ computed in the bipartition $\mathrm{\Pi}_2=(V_i,V_j)$.
In the following, we distinguish between the set of \emph{internal} border nodes  $\overrightarrow{B} := \{v \in V_B~|~\exists e \in E: \{u,v\} \subseteq e \wedge u \notin V_B\}$
and the set of \emph{external} border nodes $\overleftarrow{B} := \{u \notin V_B~|~\exists e \in E: \{u,v\} \subseteq e \wedge v \in V_B\}$.
Similarly, we distinguish between \emph{external} nets $(e \cap V_B = \emptyset)$ with no pins inside $H_B$, \emph{internal} nets $(e \cap V_B = e)$ with all pins inside $H_B$,
and \emph{border} nets $e \in \mathrm{I}(\overrightarrow{B}) \cap \mathrm{I}(\overleftarrow{B})$ with some pins inside $H_B$ and some pins outside of $H_B$.
We use $\overleftrightarrow{E_B}$ to denote the set of border nets.

A hypergraph flow problem consists of a flow network $\mathcal{N}_B=(\mathcal{V}_B,\mathcal{E}_B)$ derived from $H_B$ and two \emph{additional} nodes $\mathpzc{s}$ and $\mathpzc{t}$
that are connected to some nodes $v \in \mathcal{V}_B$. Our approach works with all flow networks presented in Section~\ref{HypergraphFlows}.
A flow problem has the cut property if the resulting min-cut bipartition $\mathrm{\Pi}_f$ of $H_B$ 
does not increase the $(\lambda -1)$-metric in $H$. Thus it has to hold that $\text{cut}(\mathrm{\Pi}_f) \leq \text{cut}(\mathrm{\Pi}_2)$.
While external nets are not affected by a max-flow computation, the max-flow min-cut theorem~\cite{ford1956maximal} ensures the cut property
for all internal nets.
Border nets however require special attention. Since a border net $e$ is only \emph{partially} contained in $H_B$,
it will remain connected to the blocks of its external border nodes in $H$. 
In case external border nodes connect $e$ to both $V_i$ and $V_j$, it will remain a cut net in $H$
even if it is removed from the cut-set in $\mathrm{\Pi}_f$. 
 It is therefore necessary to ``encode'' information about external border nodes into the flow problem.

\subparagraph*{The KaFFPa Model and its Limitations.}
In KaFFPa, this is done by directly connecting internal border nodes $\overrightarrow{B}$ to $\mathpzc{s}$ and $\mathpzc{t}$.
This approach can also be used for hypergraphs. 
In the hypergraph flow problem $\mathcal{F}_G$, the source $\mathpzc{s}$ is connected to all nodes $\mathcal{S}=\overrightarrow{B} \cap V_i$ 
and all nodes $\mathcal{T}=\overrightarrow{B} \cap V_j$ are connected to  $\mathpzc{t}$ using directed edges with infinite capacity.
While this ensures that $\mathcal{F}_G$ has the cut property, applying the graph-based model to hypergraphs unnecessarily \emph{restricts} the search space. 
Since all internal border nodes $\overrightarrow{B}$ are connected to either  $\mathpzc{s}$ or  $\mathpzc{t}$,
\emph{every} min-cut $(S,V_B \setminus S)$ will have $\mathcal{S} \subseteq S$ and $\mathcal{T} \subseteq V_B \setminus S$. 
The KaFFPa model therefore prevents \emph{all} min-cuts in which any non-cut border net (i.e., $e \in \overleftrightarrow{E_B}$ with $\lambda(e)=1$) becomes part of the cut-set.
This restricts the space of possible solutions, since corridor $B$ was computed such that \emph{even} a min-cut along either side of the border would result in a feasible cut in $H_B$.
Thus, ideally, \emph{all} vertices $v \in B$ should be able to change their block as result of an $(\mathpzc{s},\mathpzc{t})$-max-flow computation on $\mathcal{F}_G$ -- 
not only vertices $v \in B \setminus \overrightarrow{B}$.
This limitation becomes increasingly relevant for hypergraphs with large nets as well as for partitioning problems with small imbalance $\varepsilon$, since
large nets are likely to be only partially contained in $H_B$ and tight balance constraints enforce small $B$-corridors. While the former is a problem only for
HGP, the latter also applies to GP.

\subparagraph*{A more flexible Model.}
We propose a more general model that allows an $(\mathpzc{s},\mathpzc{t})$-max-flow computation to also cut through border nets 
by exploiting the structure of hypergraph flow networks. 
Instead of directly connecting $\mathpzc{s}$ and $\mathpzc{t}$ to internal border nodes $\overrightarrow{B}$ and thus preventing all min-cuts
in which these nodes switch blocks, we conceptually extend $H_B$ to contain all external border nodes $\overleftarrow{B}$ and all border nets $\overleftrightarrow{E_B}$.
The resulting hypergraph is $\overleftarrow{H_B} = (V_B \cup \overleftarrow{B}, \{e \in E~|~ e \cap V_B \neq \emptyset \})$.
The key insight now is that
by using the flow network of $\overleftarrow{H_B}$ and connecting $\mathpzc{s}$ resp. $\mathpzc{t}$ to the \emph{external} border nodes $\overleftarrow{B} \cap V_i$ resp.  $\overleftarrow{B} \cap V_j$,
we get a flow problem that does not lock \emph{any} node $v \in V_B$ in its block, since none of these nodes is directly connected to either $\mathpzc{s}$ or $\mathpzc{t}$.
Due to the max-flow min-cut theorem~\cite{ford1956maximal}, this flow problem furthermore has the cut property, since all border nets of $H_B$ are now internal nets and all external border nodes $\overleftarrow{B}$ are locked inside their block.
However, it is not necessary to use $\overleftarrow{H_B}$ instead of $H_B$ to achieve this result.
For all vertices $v \in \overleftarrow{B}$ the flow network of $\overleftarrow{H_B}$ contains paths $\langle \mathpzc{s}, v, e' \rangle$ and $\langle e'', v, \mathpzc{t}\rangle$ that only involve infinite-capacity edges.
Therefore we can remove all nodes $v \in \overleftarrow{B}$
by directly connecting $\mathpzc{s}$ and $\mathpzc{t}$ to the corresponding bridging nodes $e',e''$ via infinite-capacity edges without affecting the maximal flow~\cite{FordFulkerson}.
More precisely, in the hypergraph flow problem $\mathcal{F}_H$, we connect $\mathpzc{s}$ to all bridging nodes $e'$ corresponding to border nets $e \in \overleftrightarrow{E_B} : e \subset \overleftarrow{B} \cap V_i $ and all bridging nodes $e''$ corresponding to border nets  $e \in  \overleftrightarrow{E_B} : e \subset \overleftarrow{B} \cap V_j$ to $\mathpzc{t}$
using directed, infinite-capacity edges.
\subparagraph*{Single-Pin Border Nets.}Furthermore, we model border nets with $|e \cap \overrightarrow{B}| = 1$ more efficiently.
For such a net $e$, the flow problem contains paths of the form $\langle \mathpzc{s},e',e'',v \rangle$ or $\langle v, e', e'', \mathpzc{t} \rangle$
which can be replaced by paths of the form $\langle \mathpzc{s},e',v \rangle$ or  $\langle v, e'',\mathpzc{t} \rangle$ with
$\mathpzc{c}(e',v)=\omega(e)$ resp. $\mathpzc{c}(v,e'')=\omega(e)$. In both cases we can thus remove one bridging node and two infinite-capacity edges.
A comparison of $\mathcal{F}_G$ and $\mathcal{F}_H$ is shown in Figure~\ref{fig:balanced_bipartitioning}.

\begin{figure}[t!]
  \centering
  \includegraphics[width=\textwidth]{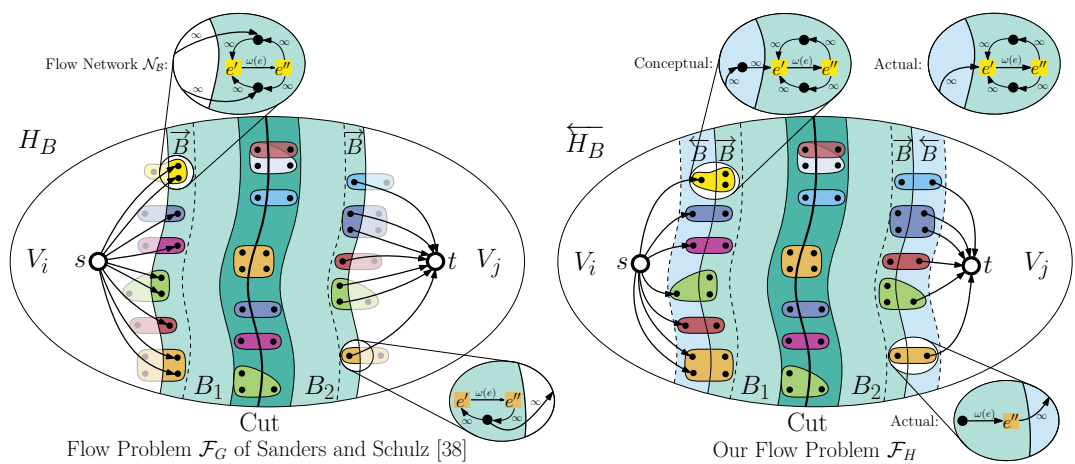}
  \caption{Comparison of the KaFFPa flow problem $\mathcal{F}_G$ and our flow problem $\mathcal{F}_H$. For clarity the zoomed in view is based on the Lawler network.}~\label{fig:balanced_bipartitioning}
\end{figure}

\subsection{Implementation Details}\label{KaHyParIntegration}
Since KaHyPar is an $n$-level partitioner, its FM-based local search algorithms are executed each time a vertex is uncontracted.
To prevent  expensive recalculations, it therefore uses a cache to maintain the gain values of FM moves throughout the $n$-level hierarchy~\cite{ahss2017alenex}.
In order to combine our flow-based refinement with FM local search, we not only perform the moves induced
by the max-flow min-cut computation but also update the FM gain cache accordingly.

Since it is not feasible to execute our algorithm on every level of the $n$-level hierarchy,
we use an exponentially spaced approach that performs flow-based refinements after uncontracting $i = 2^j$ vertices for $j \in \mathbb{N}_+ $.
This way, the algorithm is executed more often on smaller flow problems than on larger ones.
To further improve the running time, we introduce the following speedup techniques:
\begin{itemize}
\item \textbf{S1:} We modify active block scheduling such that after the first round the algorithm is only executed on a pair of blocks if at least one execution using these blocks improved connectivity or imbalance of the partition on previous levels.
\item\textbf{S2:} For all levels except the finest level: Skip flow-based refinement if the cut between two adjacent blocks is less than ten.
\item \textbf{S3:} Stop resizing the corridor $B$ if the current $(\mathpzc{s},\mathpzc{t})$-cut did not improve the previously best solution.
\end{itemize}

\section{Experimental Evaluation}\label{Experiments}
We implemented the max-flow min-cut refinement algorithm in the $n$-level hypergraph partitioning framework
\emph{KaHyPar} (\textbf{Ka}rlsruhe \textbf{Hy}pergraph \textbf{Par}titioning). 
The code is written in C++ and compiled using g++-5.2 with flags \texttt{-O3} \texttt{-march=native}.
The latest version of the framework is called KaHyPar-CA~\cite{hs2017sea}. We refer to our new algorithm as KaHyPar-MF.
Both versions use the default configuration for community-aware direct $k$-way partitioning.\footnote{\url{https://github.com/SebastianSchlag/kahypar/blob/master/config/km1_direct_kway_sea17.ini}}

\subparagraph*{Instances.}\label{Instances}
All experiments use hypergraphs from the benchmark set of Heuer and Schlag~\cite{hs2017sea}\footnote{The complete benchmark set along with detailed statistics for each
  hypergraph is publicly available from \url{http://algo2.iti.kit.edu/schlag/sea2017/}.},
which contains $488$ hypergraphs derived from four benchmark sets: the ISPD98 VLSI Circuit Benchmark Suite~\cite{ISPD98},
the DAC 2012 Routability-Driven Placement Contest~\cite{DAC2012}, the University of Florida Sparse Matrix Collection~\cite{FloridaSPM},
and the international SAT Competition 2014~\cite{SAT14Competition}.
Sparse matrices are translated into hypergraphs using the row-net model~\cite{PaToH}, i.e., each row is treated as a net and each column as a vertex.
SAT instances are converted to three different representations: For literal hypergraphs, each boolean \emph{literal} is mapped to one vertex and each clause constitutes a net~\cite{Papa2007}, while in the \emph{primal} model each variable is represented by a vertex and each clause is represented by a net.
In the \emph{dual} model the opposite is the case~\cite{FormulaPartitioning14}.
All hypergraphs have unit vertex and net weights. 

Table~\ref{tbl:benchmarkstats} gives an overview about the different benchmark sets used in the experiments. The full benchmark set is referred to as set A.
We furthermore use the representative subset of 165 hypergraphs proposed in~\cite{hs2017sea} (set B) and a smaller subset consisting of $25$ hypergraphs (set C),
which is used to devise the final configuration of KaHyPar-MF. Basic properties of set C can be found in Table~\ref{tbl:instancessmall} in Appendix~\ref{app:hypergraphs}.
Unless mentioned otherwise, all hypergraphs are partitioned into $k \in \{2,4,8,16,32,64,128\}$ blocks with $\varepsilon = 0.03$.
For each value of $k$, a $k$-way partition is considered to be \emph{one} test instance, resulting in a total of $175$ instances for set C, $1155$ instances for set B and $3416$ instances for set A.
Furthermore we use 15 graphs from \cite{DBLP:conf/wea/MeyerhenkeSS14} to compare our flow model $\mathcal{F}_H$ to the KaFFPa~\cite{kaffpa} model $\mathcal{F}_G$.
Table~\ref{tbl:graphinstances} in Appendix~\ref{app:hypergraphs} summarizes the basic properties of these graphs, which constitute set D.

\begin{table}[t!]
\centering
\caption{Overview about different benchmark sets. Set B and set C are subsets of set A.}
\label{tbl:benchmarkstats}
\begin{tabular}{lcrrrrrrrr}
  \toprule
      & Source & \#  &  DAC & ISPD98 & Primal & Dual & Literal & SPM & Graphs \\
  \midrule
Set A & \cite{hs2017sea}    & 477  & 10          & 18           & 92         & 92 & 92 & 184 & -  \\
Set B & \cite{hs2017sea}    & 165  & 5           & 10           & 30         & 30 & 30 & 60  & -   \\
Set C &  new   & 25   & -           & 5            & 5          & 5  & 5 & 5 & - \\
\midrule
Set D & \cite{DBLP:conf/wea/MeyerhenkeSS14}  & 15   & -           & -            & -          & -  & - & - & 15 \\
\bottomrule
\end{tabular}
\end{table}

\subparagraph*{System and Methodology.}\label{Methodology}
All experiments are performed on a single core of a machine consisting of two Intel Xeon E5-2670 Octa-Core processors (Sandy Bridge) 
clocked at $2.6$ GHz. The machine has $64$~GB main memory, $20$ MB L3- and 8x256 KB L2-Cache and is running RHEL 7.2.
We compare KaHyPar-MF to KaHyPar-CA, as well as to the $k$-way (hMetis-K) and the recursive bisection variant (hMetis-R) of hMetis 2.0 (p1)~\cite{hMetisRB,hMetisKway},
and to PaToH 3.2~\cite{PaToH}. These HGP libraries were chosen because they provide the best solution quality~\cite{ahss2017alenex,hs2017sea}.
The partitioning results of these tools are already available from \url{http://algo2.iti.kit.edu/schlag/sea2017/}. For each partitioner except PaToH
the results summarize ten repetitions with different seeds for each test instance and report the \emph{arithmetic 
  mean} of the computed cut and running time as well as the best cut found.
Since PaToH ignores the random seed if configured to use the quality preset, the results contain both the result of single run of the quality preset (PaToH-Q) and the average over ten repetitions using the default configuration (PaToH-D).
Each partitioner had a time limit of eight hours per test instance.
We use the same number of repetitions and the same time limit for our experiments with KaHyPar-MF.

In the following, we use the \emph{geometric mean} when averaging over different instances in order to give every instance a comparable influence on the final result.
In order to compare the algorithms in terms of solution quality, we perform a more detailed analysis using \emph{improvement plots}.
For each algorithm, these plots relate the minimum connectivity of KaHyPar-MF to the minimum connectivity produced by the corresponding algorithm on a per-instance basis.
For each algorithm, these ratios are sorted in decreasing order. The plots use a cube root scale for the y-axis to reduce right skewness~\cite{st0223}
and show the improvement of KaHyPar-MF in percent (i.e., $1-(\text{KaHyPar-MF}/\text{algorithm})$) on the y-axis.
A value below zero indicates that the partition of KaHyPar-MF was worse than the partition produced by the corresponding algorithm, while a value above zero
indicates that KaHypar-MF performed better than the algorithm in question.  A value of zero implies that the partitions of both algorithms had the same solution quality.
Values above one correspond to infeasible solutions that violated the balance constraint.
In order to include instances with a cut of zero into the results, we set the corresponding cut values to \emph{one} for ratio computations.

\begin{table}[t!]
\centering
\caption{Statistics of benchmark set B. We use $\overline{x}$ to denote mean and $\widetilde{x}$ to denote the median.}
\label{tbl:subsetstats}
\begin{tabular}{lrrrrr}
  \toprule
  Type    & \# & \multicolumn{1}{c}{$\overline{d(v)}$} &  \multicolumn{1}{c}{$\widetilde{d(v)}$} &  \multicolumn{1}{c}{$\overline{|e|}$} &  \multicolumn{1}{c}{$\widetilde{|e|}$}  \\
  \midrule
DAC     & 5  & 3.32          & 3.28           & 3.37         & 3.35  \\
ISPD    & 10 & 4.20          & 4.24           & 3.89         & 3.90  \\
Primal  & 30 & 16.29         & 9.97           & 2.63         & 2.39  \\
Literal & 30 & 8.21          & 4.99           & 2.63         & 2.39  \\
Dual    & 30 & 2.63          & 2.38           & 16.29        & 9.97  \\
SPM     & 60 & 24.78         & 14.15          & 26.58        & 15.01 \\
\bottomrule
\end{tabular}
\end{table}

\subsection{Evaluating Flow Networks, Models, and Algorithms}\label{sec:flow_algos_networks}

\subparagraph*{Flow Networks and Algorithms.}To analyze the effects of the different hypergraph flow networks we compute five bipartitions for
each hypergraph of set B with KaHyPar-CA using different seeds. Statistics of the hypergraphs are shown in Table~\ref{tbl:subsetstats}.
The bipartitions are then used to generate hypergraph flow networks for a corridor of size $|B| = \numprint{25000}$ hypernodes around the cut.
Figure~\ref{fig:flow_network_sizes} (top) summarizes the sizes of the respective flow networks in terms of number of nodes $|\mathcal{V}|$ and
number of edges $|\mathcal{E}|$ for each instance class.
 The flow networks of primal and literal SAT instances are the largest
in terms of both numbers of nodes and edges.
High average vertex degree combined with low average net sizes leads to subhypergraphs $H_B$ containing many small nets, which then induce many nodes and (infinite-capacity) edges in $\mathcal{N}_L$.
Dual instances with low average degree and large average net size on the other hand lead to smaller flow networks.
For VLSI instances (DAC, ISPD) both average degree and average net sizes are low, while for SPM~hypergraphs the opposite is the case.
This explains why SPM~flow networks have significantly more edges, despite the number of nodes being comparable in both classes.

\begin{figure}[t!]
  \centering
\begin{knitrout}
\definecolor{shadecolor}{rgb}{0.969, 0.969, 0.969}\color{fgcolor}

{\centering \includegraphics[width=\textwidth]{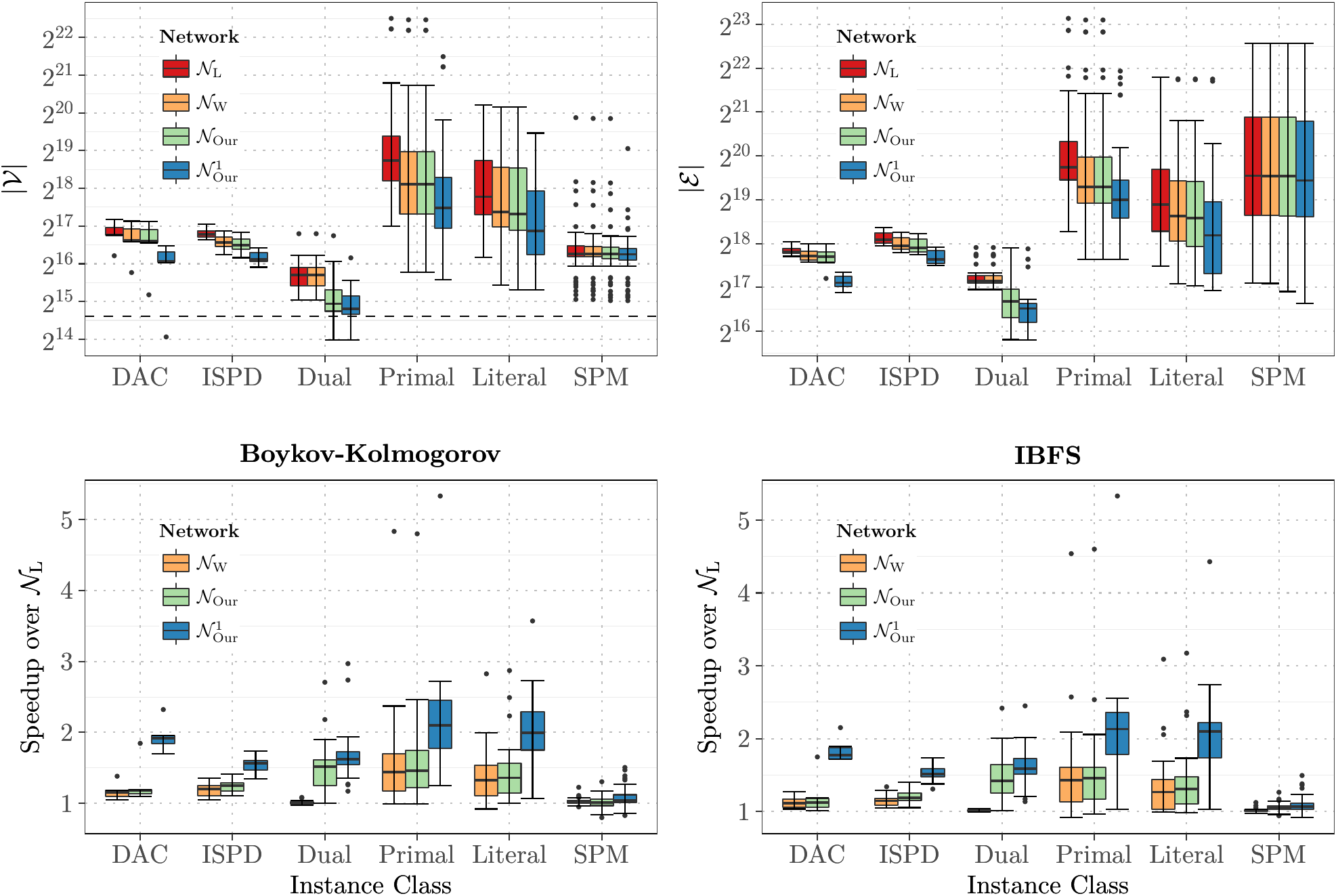} 

}

\end{knitrout}
\caption{Top: Size of the flow networks when using the Lawler network $\mathcal{N}_\text{L}$, the Liu-Wong network $\mathcal{N}_\text{W}$ and our network $\mathcal{N}_\text{Our}$.
  Network $\mathcal{N}_\text{Our}^1$ additionally models single-pin border nets more efficiently.
The dashed line indicates $\numprint{25000}$ nodes. Bottom: Speedup of BK~\cite{boykov2004experimental} and IBFS~\cite{goldberg2011maximum} max-flow algorithms over the execution on $\mathcal{N}_L$.}~\label{fig:flow_network_sizes}
\end{figure}

As expected, the Lawler-Network $\mathcal{N}_L$ induces the biggest flow problems.
Looking at the Liu-Wong network $\mathcal{N}_W$, we can see that distinguishing between graph edges and nets with $|e| \geq 3$ pins
has an effect for all hypergraphs with many small nets (i.e., DAC, ISPD, Primal, Literal).  While this technique alone does not improve
dual SAT instances, we see that the combination of the Liu-Wong approach and our removal of low degree hypernodes in  $\mathcal{N}_\text{Our}$
reduces the size of the networks for all instance classes except SPM. Both techniques only have a limited effect on these instances, since both hypernode
degrees \emph{and} net sizes are large on average. Since our flow problems are based on $B$-corridor induced subhypergraphs, $\mathcal{N}_\text{Our}^1$ additionally models single-pin border nets more efficiently as
described in Section~\ref{SourcesSinks}. This further reduces the network sizes significantly. As expected, the reduction in numbers of nodes and edges is most pronounced for hypergraphs with low average net sizes because
these instances are likely to contain many single-pin border nets.

To further see how these reductions in network size translate to improved running times of max-flow algorithms, we use these networks to create flow problems using our flow model $\mathcal{F}_H$
and compute min-cuts using two highly tuned max-flow algorithms, namely the BK-algorithm\footnote{Available from: \url{https://github.com/gerddie/maxflow}}~\cite{boykov2004experimental} and the incremental breadth-first search (IBFS) algorithm\footnote{Available from: \url{http://www.cs.tau.ac.il/~sagihed/ibfs/code.html}}~\cite{goldberg2011maximum}.
These algorithms were chosen because they performed best in preliminary experiments~\cite{MAHeuer}.
We then compare the speedups of these algorithms when executed on $\mathcal{N}_W$, $\mathcal{N}_\text{Our}$, and $\mathcal{N}_\text{Our}^1$ to the execution on the Lawler network $\mathcal{N}_L$.
As can be seen in Figure~\ref{fig:flow_network_sizes} (bottom) both algorithms benefit from improved network models and the speedups directly correlate with the reductions in network size.
While $\mathcal{N}_W$ significantly reduces the running times for \Primal~and \Literal~instances, $\mathcal{N}_\text{Our}$ additionally leads to a speedup for \Dual~instances.
By additionally considering single-pin border nets, $\mathcal{N}_\text{Our}^1$ results in an average speedup between $1.52$ and $2.21$ (except for SPM instances).
Since IBFS outperformed the BK algorithm in \cite{MAHeuer}, we use $\mathcal{N}_\text{Our}^1$ and IBFS in all following experiments.

\begin{table}
  \caption{Comparing the KaFFPa flow model $\mathcal{F}_G$ with our model $\mathcal{F}_H$ as described in Section~\ref{SourcesSinks}.
    The table shows the average improvement of $\mathcal{F}_H$ over $\mathcal{F}_G$  (in [\%]) for different imbalance parameters $\varepsilon$ on hypergraphs as well as on graphs.
    All experiments use configuration \FlowVariant{+}{-}{-}.}\label{tbl:flow_models}
  \centering
  
\begin{tabular}{r*{3}{S[table-number-alignment = center]}
                    *{3}{S[table-number-alignment = center]}}
  \toprule
 & \multicolumn{3}{c}{\textsc{Hypergraphs}} & \multicolumn{3}{c}{\textsc{Graphs}} \\
$\alpha'$ & \multicolumn{1}{c}{$\varepsilon = 1\%$}  & \multicolumn{1}{c}{$\varepsilon = 3\%$}  & \multicolumn{1}{c}{$\varepsilon = 5\%$}  & \multicolumn{1}{c}{$\varepsilon = 1\%$}  & \multicolumn{1}{c}{$\varepsilon = 3\%$}  & \multicolumn{1}{c}{$\varepsilon = 5\%$} \\     
\midrule%
\csname @@input\endcsname experiments/flow_alpha/flow_alpha_modeling_comparison_epsilon_table.tex
\bottomrule
\end{tabular}
\label{tbl:alpha_comparison_exp}
\end{table}

\subparagraph*{Flow Models.} We now compare the flow model $\mathcal{F}_G$ of KaFFPa to our advanced model $\mathcal{F}_H$ described in Section~\ref{SourcesSinks}.
The experiments summarized in Table~\ref{tbl:flow_models} were performed using sets C and D. To focus on the impact of the models
on solution quality, we deactivated KaHyPar's FM local search algorithms and only use flow-based refinement without the most balanced minimum cut heuristic.
The results confirm our hypothesis that $\mathcal{F}_G$ restricts the space of possible solutions.
For \emph{all} flow problem sizes and all imbalances tested, $\mathcal{F}_H$ yields better solution quality. As expected, the effects are most
pronounced for small flow problems and small imbalances where many vertices are likely to be border nodes.
Since these nodes are locked inside their respective block in $\mathcal{F}_G$, they prevent all non-cut border nets from becoming part of the cut-set.
Our model, on the other hand, allows \emph{all} min-cuts that yield a feasible solution for the original partitioning problem.
The fact that this effect \emph{also} occurs for the graphs of set D indicates that our model can also be effective
for traditional graph partitioning. All following experiments are performed using $\mathcal{F}_H$.

\subsection{Configuring the Algorithm}\label{sec:algo_configuration}
We now evaluate different configurations of the max-flow min-cut based refinement framework on set C.
In the following, KaHyPar-CA~\cite{hs2017sea} is used as a  reference. Since it neither uses (F)lows nor the (M)ost balanced minimum cut heuristic and
only relies on the (FM) algorithm for local search, it is referred to as \FlowVariant{-}{-}{+}. This basic configuration is then successively extended with specific components. The results of our
experiments are summarized in Table~\ref{tbl:alpha_exp} for increasing scaling parameter $\alpha'$.
The table furthermore includes a configuration \Constant{128}. In this configuration all components are enabled \FlowVariant{+}{+}{+} and
we perform flow-based refinements every 128 uncontractions. While this configuration is slow, it is used as a reference point for the quality achievable
using flow-based refinement.

\begin{table}[!htb]
  \caption{Different configurations of our flow-based refinement
          framework for increasing $\alpha'$. Column Avg$[\%]$ reports the quality improvement relative to the reference configuration \FlowVariant{-}{-}{+}.}
\label{tbl:alpha_exp}
\centering
\begin{tabular*}{\columnwidth}{@{\extracolsep{\fill}}r*{2}{S[table-number-alignment = center]}
                    *{2}{S[table-number-alignment = center,table-column-width=0.75cm]}
                    *{2}{S[table-number-alignment = center,table-column-width=0.75cm]}
                    *{2}{S[table-number-alignment = center,table-column-width=0.75cm]}
                    *{2}{S[table-number-alignment = center,table-column-width=0.75cm]}}
\toprule
 & \multicolumn{2}{c}{\FlowVariant{+}{-}{-}} & \multicolumn{2}{c}{\FlowVariant{+}{+}{-}}  & \multicolumn{2}{c}{\FlowVariant{+}{-}{+}} & \multicolumn{2}{c}{\FlowVariant{+}{+}{+}} & \multicolumn{2}{c}{\Constant{128}} \\
$\alpha'$ & Avg $[\%]$ & $t[s]$ & Avg $[\%]$ & $t[s]$ & Avg $[\%]$ & $t[s]$ & Avg $[\%]$ & $t[s]$ & Avg $[\%]$ & $t[s]$ \\
\midrule%
\csname @@input\endcsname experiments/flow_alpha/flow_alpha_table_m2_ibfs.tex
\bottomrule
\end{tabular*}
\end{table}

The results indicate that only using flows \FlowVariant{+}{-}{-} as refinement technique is inferior to localized FM local search in regard to both
running time and solution quality. Although the quality improves with increasing flow problem size (i.e., increasing $\alpha'$), the average connectivity is still worse than the reference configuration.
Enabling the most balanced minimum cut heuristic improves partitioning quality.
Configuration \FlowVariant{+}{+}{-} performs better than the basic configuration for  $\alpha' \geq 8$. 
By combining flows with the FM algorithm \FlowVariant{+}{-}{+} we get a configuration that improves upon the baseline configuration
even for small flow problems. However, comparing this variant with \FlowVariant{+}{+}{-} for $\alpha' = 16$, we see that using large flow problems together with
the most balanced minimum cut heuristic yields solutions of comparable quality. Enabling all components \FlowVariant{+}{+}{+} and using large flow problems
performs best. Furthermore we see that enabling FM local search slightly improves the running time for $\alpha' \geq 8$.
This can be explained by the fact that the FM algorithm already produces good cuts between the blocks such that fewer rounds of
pairwise flow refinements are necessary to further improve the solution.
Comparing configuration \FlowVariant{+}{+}{+} with \Constant{128} shows that performing flows more often further improves solution quality at the cost
of slowing down the algorithm by more than an order of magnitude.
In all further experiments, we therefore use configuration \FlowVariant{+}{+}{+} with $\alpha' = 16$ for KaHyPar-MF.
This configuration also performed best in the effectiveness tests presented in Appendix~\ref{app:effective}.
While this configuration performs better than KaHyPar-CA, its running time is still more than a factor of $3$ higher.

We therefore perform additional experiments on set B and successively enable the speedup heuristics described in Section~\ref{KaHyParIntegration}.
The results are summarized in Table~\ref{tbl:heuristics}. Only executing pairwise flow refinements on blocks that lead to an improvement on previous levels (S1)
reduces the running time of flow-based refinement by a factor of $1.27$, while skipping flows in case of small cuts (S2) results in a further speedup of $1.19$.
By additionally stopping the resizing of the flow problem as early as possible (S3), we decrease the running time of flow-based improvement by a factor of $2$ in total,
while still computing solutions of comparable quality. Thus in the comparisons with other systems, all heuristics are enabled.

\begin{table}[t!]
  \caption{Comparison of quality improvement and running times using speedup heuristics. Column $t_{\text{flow}}[s]$ refers to the running time of flow-based refinement, column $t[s]$ to the
    total partitioning time.}
\label{tbl:heuristics}
\centering
\begin{tabular}{l*{4}{S[table-number-alignment = center,table-column-width=1cm]}}
  \toprule
Configuration & \multicolumn{1}{c}{Avg $[\%]$} & \multicolumn{1}{c}{Min $[\%]$} & \multicolumn{1}{c}{$t_{\text{flow}}[s]$} & \multicolumn{1}{c}{$t[s]$} \\ 
\midrule%
\csname @@input\endcsname experiments/speed_up_heuristics/heuristic_table_m2_ibfs.tex
\bottomrule
\end{tabular} 
\end{table}

\subsection{Comparison with other Systems}\label{sec:final_evaluation}
Finally, we compare \KaHyPar{MF} to different state-of-the-art hypergraph
partitioners on the full benchmark set. We exclude the same $194$ out of $3416$ instances as in \cite{hs2017sea} because either 
\PaToH{Q} could not allocate enough memory or other partitioners did not finish in time. The
excluded instances are shown in Table~\ref{tbl:excluded} in Appendix~\ref{app:excluded_instances}. Note that KaHyPar-MF did not
lead to any further exclusions. The following comparison is therefore based on the remaining 3222 instances.
As can be seen in Figure~\ref{fig:final_flow}, KaHyPar-MF outperforms \emph{all} other algorithms on \emph{all} benchmark sets.
Comparing the best solutions of \KaHyPar{MF} to each partitioner individually across all instances (top left), \KaHyPar{MF} produced \emph{better} partitions than
\PaToH{Q}, \PaToH{D}, \hMetis{K}, \KaHyPar{CA}, \hMetis{R} for $92.1\%$, $91.7\%$, $85.1\%$, $83.7\%$, and $75.6\%$ of the instances, respectively.

Comparing the best solutions of all systems simultaneously, \KaHyPar{MF} produced the best partitions for $2427$ of
 the $3222$ instances. It is followed by \hMetis{R} ($678$), \KaHyPar{CA} ($388$), \hMetis{K} ($352$),  
 \PaToH{D} ($154$), and \PaToH{Q} ($146$). Note that for some instances multiple partitioners computed the same best solution and that we disqualified
 infeasible solutions that violated the balance constraint.

Figure~\ref{fig:final_flow_k} shows that \KaHyPar{MF} also performs best for different values of $k$ and that pairwise flow refinements
are an effective strategy to improve $k$-way partitions. As can be seen in Table~\ref{tbl:mincon_solution_quality}, the improvement over
\KaHyPar{CA} is most pronounced for hypergraphs derived from matrices of web graphs and social networks\footnote{Based on the following matrices: \texttt{webbase-1M}, \texttt{ca-CondMat}, \texttt{soc-sign-epinions}, \texttt{wb-edu}, \texttt{IMDB}, \texttt{as-22july06}, \texttt{as-caida}, \texttt{astro-ph}, \texttt{HEP-th}, \texttt{Oregon-1}, \texttt{Reuters911},
  \texttt{PGPgiantcompo}, \texttt{NotreDame\_www}, \texttt{NotreDame\_actors}, \texttt{p2p-Gnutella25}, \texttt{Stanford}, \texttt{cnr-2000}.} and dual SAT instances.
While the former are difficult to partition due to skewed degree and net size distributions, the latter are difficult because they contain many large nets.

Finally, Table~\ref{tbl:running_time} compares the running times of all partitioners. By using simplified flow networks, highly tuned flow algorithms and several
techniques to speed up the flow-based refinement framework, \KaHyPar{MF} is less than a factor of two slower than \KaHyPar{CA} and still achieves a running time comparable
to that of hMetis.

\begin{table}[t!h]
  \caption{Comparing the best solutions of \KaHyPar{MF} with the best results of \KaHyPar{CA} and
    other partitioners for different benchmark sets (top) and different values of $k$ (bottom). All values correspond to the quality improvement
    of \KaHyPar{MF} relative to the respective partitioner  (in \%).}
\label{tbl:mincon_solution_quality} 
\centering
\begin{tabular}{c*{8}{S[table-number-alignment = center]}}
  \toprule
\multirow{2}{*}{Algorithm} & \multicolumn{7}{c}{Min. $(\lambda - 1)$} \\
\cmidrule{2-9}
 & \multicolumn{1}{c}{\ALL} & \multicolumn{1}{c}{\DAC} & \multicolumn{1}{c}{\ISPD} & \multicolumn{1}{c}{\Primal} & \multicolumn{1}{c}{\Literal} & \multicolumn{1}{c}{\Dual} & \multicolumn{1}{c}{\SPM} &  \multicolumn{1}{c}{WebSoc} \\
\midrule%
\csname @@input\endcsname experiments/final_flow/final_flow_min_km1_per_instance_m2_ibfs.tex
\midrule%
& \multicolumn{1}{c}{$k = 2$} & \multicolumn{1}{c}{$k = 4$} & \multicolumn{1}{c}{$k = 8$} & \multicolumn{1}{c}{$k = 16$} & \multicolumn{1}{c}{$k = 32$} & \multicolumn{1}{c}{$k = 64$} & \multicolumn{1}{c}{$k = 128$}  \\
\midrule%
\csname @@input\endcsname experiments/final_flow/final_flow_min_km1_per_k_m2_ibfs.tex
\bottomrule
\end{tabular} 
\end{table}

\begin{table}[ht!]
  \caption{Comparing the average running times of \KaHyPar{MF} with \KaHyPar{CA} and
         other hypergraph partitioners for different benchmark sets (top) and different values of $k$ (bottom).}
\label{tbl:running_time} 
\centering
\begin{tabular}{c*{8}{S[table-number-alignment = center]}}
  \toprule
\multirow{2}{*}{Algorithm} & \multicolumn{7}{c}{Running Time $t[s]$} \\
\cmidrule{2-9}
 & \multicolumn{1}{c}{\ALL} & \multicolumn{1}{c}{\DAC} & \multicolumn{1}{c}{\ISPD} & \multicolumn{1}{c}{\Primal} & \multicolumn{1}{c}{\Literal} & \multicolumn{1}{c}{\Dual} & \multicolumn{1}{c}{\SPM} & \multicolumn{1}{c}{WebSoc}  \\
\midrule%
\csname @@input\endcsname experiments/final_flow/final_flow_running_time_m2_ibfs.tex
\midrule%
 & \multicolumn{1}{c}{$k = 2$} & \multicolumn{1}{c}{$k = 4$} & \multicolumn{1}{c}{$k = 8$} & \multicolumn{1}{c}{$k = 16$} & \multicolumn{1}{c}{$k = 32$} & \multicolumn{1}{c}{$k = 64$} & \multicolumn{1}{c}{$k = 128$}  \\
\midrule%
\csname @@input\endcsname experiments/final_flow/final_flow_running_time_per_k_m2_ibfs.tex
\bottomrule
\end{tabular} 
\end{table}

\begin{figure}
\centering
\hspace*{-1cm}
\includegraphics[width=\textwidth]{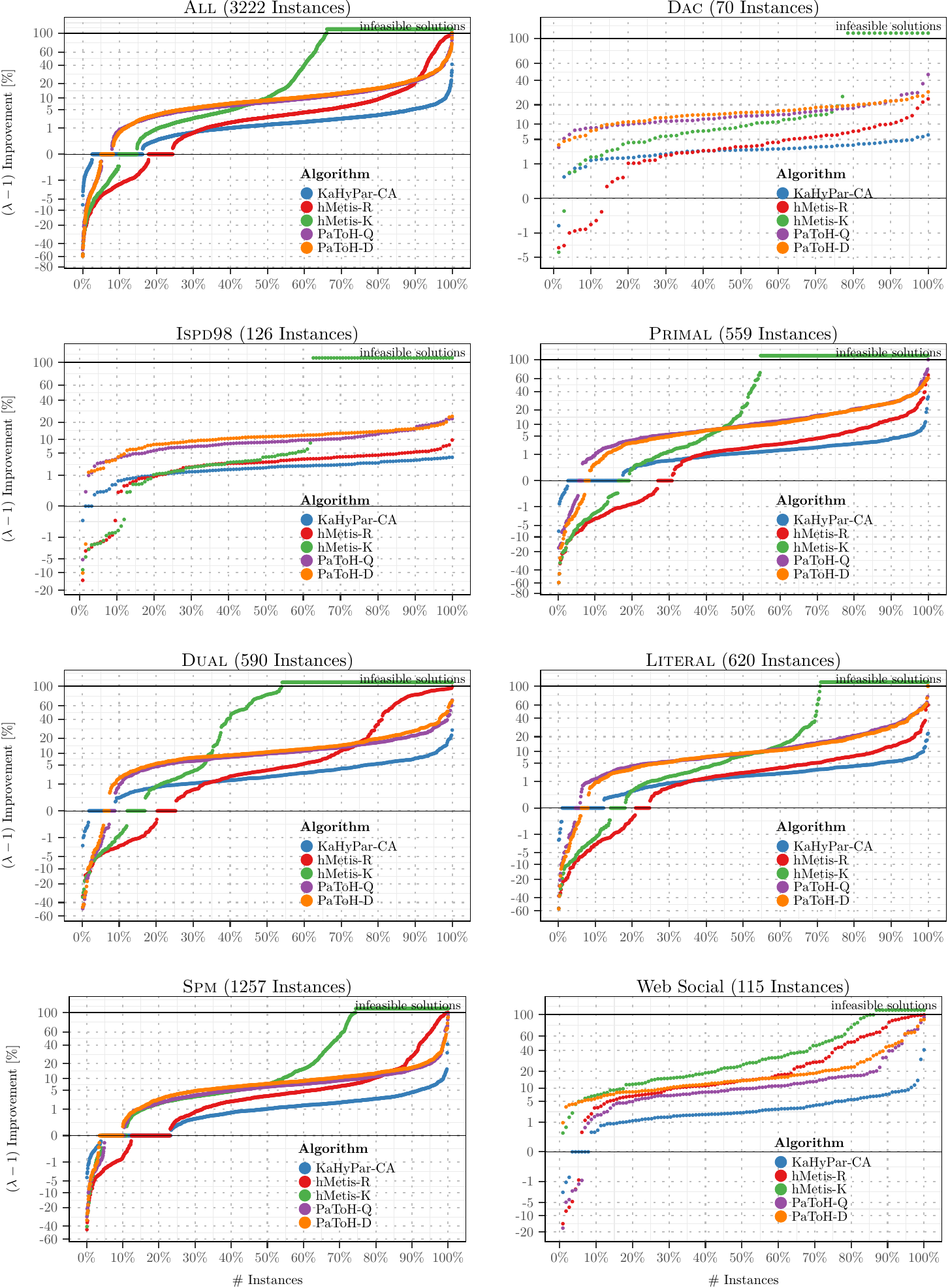} %
\caption{Min-Cut improvement plots comparing \KaHyPar{MF} with \KaHyPar{CA} and
         other partitioners for different instance classes.}
\label{fig:final_flow} 
\end{figure}

\begin{figure}
\centering
\hspace*{-1cm}
\includegraphics[width=\textwidth]{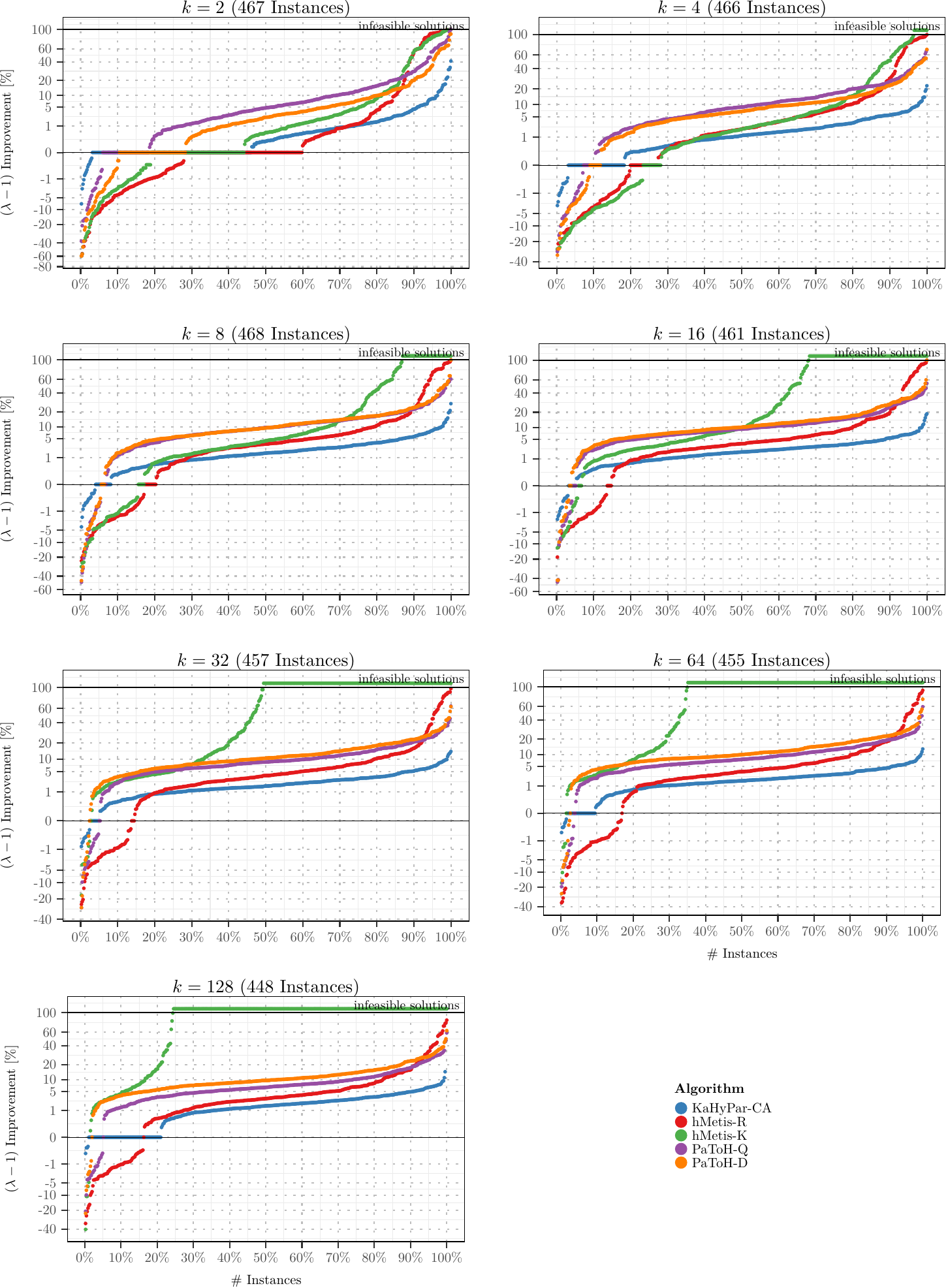} %
\caption{Min-Cut improvement plots comparing \KaHyPar{MF} with \KaHyPar{CA} and
         other partitioners for different values of $k$.}
\label{fig:final_flow_k}
\end{figure} 
\clearpage
\section{Conclusion}
We generalize the flow-based refinement framework of KaFFPa~\cite{kaffpa} from graph to hypergraph partitioning.
We reduce the size of Liu and Wong's hypergraph flow network~\cite{GraphEdgeReduction} by removing low degree hypernodes
and exploiting the fact that our flow problems are built on subhypergraphs of the input hypergraph.
Furthermore we identify shortcomings of the KaFFPa~\cite{kaffpa} approach that restrict the search space of feasible solutions significantly and
introduce an advanced model that overcomes these limitations by exploiting the structure of hypergraph flow networks.
Lastly, we present techniques to improve the running time of the flow-based refinement framework by a factor of $2$ without affecting solution quality.
The resulting hypergraph partitioner KaHyPar-MF performs better than \emph{all} competing algorithms on \emph{all} instance classes of a large benchmark set
and still has a running time comparable to that of hMetis.

Since our flow problem formulation yields significantly better solutions for both hypergraphs \emph{and} graphs than the KaFFPa~\cite{kaffpa} approach,
future work includes the integration of our flow model into KaFFPa and the evaluation in the context of a high quality graph partitioner.
Furthermore an approach similar to Yang and Wong~\cite{yang1996balanced} could be used as an alternative to the most balanced minimum cut heuristic and adaptive $B$-corridor resizing.
We also plan to extend our framework to optimize other objective functions such as cut or sum of external degrees.

\clearpage
\bibliography{p50-schlag.bib}
\newpage
\appendix
\section{Effectiveness Tests}\label{app:effective}
To evaluate the effectiveness of our configurations presented in Section \ref{sec:algo_configuration}
we give each configuration the \emph{same} time to compute a partition. For each instance (hypergraph, $k$),
we execute each configuration once and note the \emph{largest} running time $t_{H,k}$.
Then each configuration gets time $3 t_{H,k}$ to compute a partition (i.e., we take the best partition out of
several repeated runs). Whenever a new run of a partition would exceed the largest running time,
we perform the next run with a certain probability such that the expected running time is $3t_{H,k}$.
The results of this procedure, which was initially proposed in \cite{kaffpa}, are presented
in Table~\ref{tbl:alpha_effectiveness_exp}. We see that the combinations of flow-based refinement and FM local search
perform better than repeated executions of the baseline configuration \FlowVariant{-}{-}{+}.
The most effective configuration is \FlowVariant{+}{+}{-} with $\alpha'=16$, which was chosen as the default
configuration for \KaHyPar{MF}.

\begin{table}[ht]
  \caption{Results of the effectiveness test 
          for different configurations of our flow-based refinement
          framework for increasing $\alpha'$. The quality in column Avg$[\%]$ is relative
          to the baseline configuration \FlowVariant{-}{-}{+}. }
\centering
\begin{tabular}{rS[table-number-alignment = center]
                   S[table-number-alignment = center]
                   S[table-number-alignment = center]
                   S[table-number-alignment = center]}
  \toprule
Config &  \multicolumn{1}{c}{\FlowVariant{+}{-}{-}} &  \multicolumn{1}{c}{\FlowVariant{+}{+}{-}}  &  \multicolumn{1}{c}{\FlowVariant{+}{-}{+}}  &  \multicolumn{1}{c}{\FlowVariant{+}{+}{+}} \\
$\alpha'$ & Avg $[\%]$ & Avg $[\%]$ & Avg $[\%]$ & Avg $[\%]$ \\
\midrule%
\csname @@input\endcsname experiments/flow_alpha/flow_alpha_effectiveness_table_m2_ibfs2.tex
\bottomrule
\end{tabular}

\label{tbl:alpha_effectiveness_exp}
\end{table}
\clearpage

\section{Average Connectivity Improvement}\label{app:avg_improvement}
\begin{table}[ht!]
  \caption{Comparing the average solution quality of \KaHyPar{MF} with the average results of \KaHyPar{CA} and
    other partitioners for different benchmark sets (top) and different values of $k$ (bottom). All values correspond to the quality improvement
    of \KaHyPar{MF} relative to the respective partitioner (in \%).}
\label{tbl:running_time} 
\centering
\begin{tabular}{c*{8}{S[table-number-alignment = center]}}
  \toprule
\multirow{2}{*}{Algorithm} & \multicolumn{7}{c}{Avg. $(\lambda - 1)$} \\
\cmidrule{2-9}
 & \multicolumn{1}{c}{\ALL} & \multicolumn{1}{c}{\DAC} & \multicolumn{1}{c}{\ISPD} & \multicolumn{1}{c}{\Primal} & \multicolumn{1}{c}{\Literal} & \multicolumn{1}{c}{\Dual} & \multicolumn{1}{c}{\SPM} &  \multicolumn{1}{c}{WebSoc}  \\
\midrule%
\csname @@input\endcsname experiments/final_flow/final_flow_km1_per_instance_m2_ibfs.tex
\midrule
& \multicolumn{1}{c}{$k = 2$} & \multicolumn{1}{c}{$k = 4$} & \multicolumn{1}{c}{$k = 8$} & \multicolumn{1}{c}{$k = 16$} & \multicolumn{1}{c}{$k = 32$} & \multicolumn{1}{c}{$k = 64$} & \multicolumn{1}{c}{$k = 128$}  \\
\midrule%
\csname @@input\endcsname experiments/final_flow/final_flow_km1_per_k_m2_ibfs.tex
\bottomrule
\end{tabular} 
\end{table}

\clearpage

\section{Properties of Benchmark Sets}\label{app:hypergraphs}
\begin{table}[h]
        \caption{Basic properties of our parameter tuning benchmark set. The number of pins is denoted with $p$.}\label{tbl:instancessmall}

\centering
\begin{tabular}{llrrr}
  \toprule
Class & Hypergraph                & \multicolumn{1}{c}{$n$}     & \multicolumn{1}{c}{$m$}    & \multicolumn{1}{c}{$p$}     \\
\midrule
\multirow{5}{*}{ISPD} & ibm06                     & \numprint{32498}   & \numprint{34826}  & \numprint{128182}  \\
&ibm07                     & \numprint{45926}   & \numprint{48117}  & \numprint{175639}  \\ 
&ibm08                     & \numprint{51309}   & \numprint{50513}  & \numprint{204890}  \\ 
&ibm09                     & \numprint{53395}   & \numprint{60902}  & \numprint{222088}  \\ 
&ibm10                     & \numprint{69429}   & \numprint{75196}  & \numprint{297567}  \\
\midrule
\multirow{5}{*}{Dual} &6s9                       & \numprint{100384} & \numprint{34317} & \numprint{234228} \\ 
&6s133                     & \numprint{140968} & \numprint{48215} & \numprint{328924} \\
&6s153                     & \numprint{245440} & \numprint{85646} & \numprint{572692} \\
&dated-10-11-u             & \numprint{629461} & \numprint{141860} & \numprint{1429872} \\ 
&dated-10-17-u             & \numprint{1070757} & \numprint{229544} & \numprint{2471122} \\ 
\midrule
\multirow{5}{*}{Literal}&6s133                     & \numprint{96430}  & \numprint{140968} & \numprint{328924} \\
&6s153                     & \numprint{171292}  & \numprint{245440} & \numprint{572692} \\
&aaai10-planning-ipc5      & \numprint{107838}  & \numprint{308235} & \numprint{690466} \\
&dated-10-11-u             & \numprint{283720}  & \numprint{629461} & \numprint{1429872} \\
&atco\_enc2\_opt1\_05\_21  & \numprint{112732}  & \numprint{526872} & \numprint{2097393} \\
\midrule
\multirow{5}{*}{Primal}   & 6s153                    & \numprint{85646}  & \numprint{245440} & \numprint{572692} \\
 &aaai10-planning-ipc5     & \numprint{53919}  & \numprint{308235} & \numprint{690466} \\
 &hwmcc10-timeframe        & \numprint{163622}  & \numprint{488120} & \numprint{1138944} \\
 &dated-10-11-u            & \numprint{141860}  & \numprint{629461} & \numprint{1429872} \\
&atco\_enc2\_opt1\_05\_21 & \numprint{56533}  & \numprint{526872} & \numprint{2097393} \\
\midrule
\multirow{5}{*}{SPM}& mult\_dcop\_01           & \numprint{25187}   & \numprint{25187}  & \numprint{193276} \\
 &vibrobox                 & \numprint{12328}   & \numprint{12328}  & \numprint{342828} \\
 &RFdevice                 & \numprint{74104}   & \numprint{74104}  & \numprint{365580} \\
 &mixtank\_new             & \numprint{29957}   & \numprint{29957}  & \numprint{1995041} \\
 &laminar\_duct3D          & \numprint{67173}   & \numprint{67173}  & \numprint{3833077} \\
 \bottomrule
        \end{tabular}
        \end{table}

\begin{table}
  \caption{Basic properties of the graph instances.}\label{tbl:graphinstances}

\centering
\begin{tabular}{lrr}
  \toprule
  Graph                & \multicolumn{1}{c}{$n$}     & \multicolumn{1}{c}{$m$}                                                                                                                                                                                                                                                                                            \\
  \cline{1-3}
  p2p-Gnutella04       & \numprint{6405}   & \numprint{29215}                                                                                                                                                                                                                                                                                                             \\
  wordassociation-2011 & \numprint{10617}  & \numprint{63788}                                                                                                                                                                                                                                                                                                             \\
  PGPgiantcompo        & \numprint{10680}  & \numprint{24316}                                                                                                                                                                                                                                                                                                             \\
  email-EuAll          & \numprint{16805}  & \numprint{60260}                                                                                                                                                                                                                                                                                                             \\
  as-22july06          & \numprint{22963}  & \numprint{48436}                                                                                                                                                                                                                                                                                                             \\
  soc-Slashdot0902     & \numprint{28550}  & \numprint{379445}                                                                                                                                                                                                                                                                                                            \\
  loc-brightkite       & \numprint{56739}  & \numprint{212945}                                                                                                                                                                                                                                                                                                            \\
  enron                & \numprint{69244}  & \numprint{254449}                                                                                                                                                                                                                                                                                                            \\
  loc-gowalla          & \numprint{196591} & \numprint{950327}                                                                                                                                                                                                                                                                                                            \\
  coAuthorsCiteseer    & \numprint{227320} & \numprint{814134}                                                                                                                                                                                                                                                                                                            \\
  wiki-Talk            & \numprint{232314} & $\approx$1.5M                                                                                                                                                                                                                                                                                                                \\
  citationCiteseer & \numprint{268495} & $\approx$1.2M                                                                                                                                                                                                                                                                                                                    \\
  coAuthorsDBLP    & \numprint{299067} & \numprint{977676}                                                                                                                                                                                                                                                                                                                \\
  cnr-2000         & \numprint{325557} & $\approx$2.7M                                                                                                                                                                                                                                                                                                                    \\
  web-Google       & \numprint{356648} & $\approx$2.1M                                                                                                                                                                                                                                                                                                                    \\
  \bottomrule
\end{tabular}
\end{table}
\clearpage
\section{Excluded Instances}\label{app:excluded_instances}
{\footnotesize
  \begin{longtable}{p{5cm}lllllll}
    \caption{Instances excluded from the full benchmark set evaluation. Note that using flow-based refinements did not lead to any further exclusions.}\label{tbl:excluded}                                                                                                                                                                                                                                                                                                              \\
       \toprule
Hypergraph                                                                                                                    & 2           & 4                               & 8                               & 16                              & 32                              & 64                                       & 128                                      \\
 \midrule
\Primal \\
\midrule
10pipe-q0-k                                                                                                          & $\square$   & $\square$                       & $\square$                       & $\square$                       & $\square$                       & $\square$                                & $\square$                                \\
11pipe-k                                                                                                             & $\square$   & $\square$                       & $\square$                       & $\square$                       & $\square$                       & $\square$                                & \ding{109}$\square$                      \\
  
  11pipe-q0-k                                                                                                          & $\square$   & $\square$                       & $\square$                       & $\square$                       & $\square$                       & $\square$                                & $\square$                                \\
  
  9dlx-vliw-at-b-iq3                                                                                                   & $\square$   & $\square$                       & $\square$                       & $\square$                       & $\square$                       & $\square$                                & $\square$                                \\
  9vliw-m-9stages-iq3-C1-bug7                                                                                          & $\triangle$ & $\triangle$                     &                                 & $\triangle$                     & \ding{109}$\triangle$           & \ding{109}$\triangle$                    & \ding{109}$\triangle$                    \\
  9vliw-m-9stages-iq3-C1-bug8                                                                                          & $\triangle$ & $\triangle$                     &                                 & $\triangle$                     & \ding{109}$\triangle$           & \ding{109}$\triangle$                    & \ding{109}$\triangle$                    \\
  blocks-blocks-37-1.130-NOTKNOWN                                                                                      & $\square$   & $\square$                       & $\square$                       & $\square$                       & $\square$                       & $\square$                                & $\square$                                \\
  openstacks-p30-3.085-SAT                                                                                             & $\square$   & $\square$                       & $\square$                       & $\square$                       & $\square$                       & $\square$                                & $\square$                                \\
  openstacks-sequencedstrips-nonadl-nonnegated-os-sequencedstrips-p30-3.025-NOTKNOWN                                   & $\square$   & $\square$                       & $\square$                       & $\square$                       & $\square$                       & $\square$                                & $\square$                                \\
  openstacks-sequencedstrips-nonadl-nonnegated-os-sequencedstrips-p30-3.085-SAT                                        & $\square$   & $\square$                       & $\square$                       & $\square$                       & $\square$                       & $\square$                                & $\square$                                \\
transport-transport-city-sequential-25nodes-1000size-3degree-100mindistance-3trucks-10packages-2008seed.050-NOTKNOWN & $\square$   &                                 &                                 & $\square$                       &                                 &                                          & $\square$                                \\

  velev-vliw-uns-2.0-uq5                                                                                               & $\square$   & $\square$                       & $\square$                       & $\square$                       & $\square$                       & $\square$                                & $\square$                                \\

  velev-vliw-uns-4.0-9                                                                                                 & $\square$   & $\square$                       & $\square$                       & $\square$                       & $\square$                       & $\square$                                & $\square$                                \\
  \midrule
\Literal \\
\midrule
  11pipe-k                                                                                                                    &             &                                 &                                 & \ding{109}                      & \ding{109}                      & \ding{109}                               & \ding{109}                               \\
  9vliw-m-9stages-iq3-C1-bug7                                                                                                 & $\triangle$ & $\triangle$                     & \ding{108}\ding{109}$\triangle$ & \ding{108}\ding{109}$\triangle$ & \ding{108}\ding{109}$\triangle$ & \ding{108}\ding{109}$\square$$\triangle$ & \ding{108}\ding{109}$\square$$\triangle$ \\

  9vliw-m-9stages-iq3-C1-bug8                                                                                                 & $\triangle$ & $\triangle$                     & \ding{108}\ding{109}$\triangle$ & \ding{108}\ding{109}$\triangle$ & \ding{108}\ding{109}$\triangle$ & \ding{108}\ding{109}$\square$$\triangle$ & \ding{108}\ding{109}$\square$$\triangle$ \\

  blocks-blocks-37-1.130                                                                                             &             & $\square$                       & $\square$                       & $\square$                       & $\square$                       & $\square$                                & $\square$                                \\
  \midrule
  \Dual \\
\midrule
10pipe-q0-k                                                                                                            &             &                                 &                                 & $\triangle$                     & $\triangle$                     & $\triangle$                              & \ding{109}$\triangle$                    \\
11pipe-k                                                                                                               & $\triangle$ & \ding{109}$\triangle$           & \ding{109}$\triangle$           & \ding{109}$\triangle$           & \ding{109}$\triangle$           & \ding{109}$\triangle$                    & \ding{109}$\triangle$                    \\
11pipe-q0-k                                                                                                            &             &                                 &                                 &                                 & $\triangle$                     & \ding{109}$\triangle$                    & \ding{109}$\triangle$                    \\
9dlx-vliw-at-b-iq3                                                                                                     &             &                                 &                                 &                                 &                                 &                                          & $\triangle$                              \\
9vliw-m-9stages-iq3-C1-bug7                                                                                            & $\triangle$ & \ding{108}\ding{109}$\triangle$ & \ding{108}\ding{109}$\triangle$ & \ding{108}\ding{109}$\triangle$ & \ding{108}\ding{109}$\triangle$ & \ding{108}\ding{109}$\triangle$          & \ding{108}\ding{109}$\triangle$          \\
9vliw-m-9stages-iq3-C1-bug8                                                                                            & $\triangle$ & \ding{108}\ding{109}$\triangle$ & \ding{108}\ding{109}$\triangle$ & \ding{108}\ding{109}$\triangle$ & \ding{108}\ding{109}$\triangle$ & \ding{108}\ding{109}$\triangle$          & \ding{108}\ding{109}$\triangle$          \\
blocks-blocks-37-1.130-NOTKNOWN                                                                                        &             & \ding{109}                      & \ding{108}\ding{109}            & \ding{108}\ding{109}            & \ding{108}\ding{109}            & \ding{108}\ding{109}                     & \ding{108}\ding{109}$\triangle$          \\
E02F20                                                                                                                 &             &                                 &                                 &                                 &                                 &                                          & \ding{109}                               \\
E02F22                                                                                                                 &             &                                 &                                 &                                 &                                 & \ding{109}                               & \ding{109}                               \\
  q-query-3-L100-coli.sat                                                                                                &             &                                 &                                 &                                 &                                 &                                          & $\triangle$                              \\
  q-query-3-L150-coli.sat                                                                                                &             &                                 &                                 &                                 &                                 & $\triangle$                              & $\triangle$                              \\
  q-query-3-L200-coli.sat                                                                                                &             &                                 &                                 &                                 & $\triangle$                     & $\triangle$                              & $\triangle$                              \\
  q-query-3-L80-coli.sat                                                                                                 &             &                                 &                                 &                                 &                                 &                                          & $\triangle$                              \\
transport-transport-city-sequential-25nodes-1000size-3degree-100mindistance-3trucks-10packages-2008seed.030-NOTKNOWN   &             &                                 &                                 &                                 &                                 &                                          & $\triangle$                              \\
  velev-vliw-uns-2.0-uq5                                                                                                &             &                                 & $\triangle$                     & $\triangle$                     & $\triangle$                     & $\triangle$                              & $\triangle$                              \\
  velev-vliw-uns-4.0-9                                                                                                   &             &                                 &                                 &                                 & $\triangle$                     & $\triangle$                              & $\triangle$                              \\
\\
  \midrule
\SPM \\
\midrule
  192bit                                                                                                                      & $\square$   &                                 &                                 & $\square$                       &                                 &                                          &                                          \\
  appu                                                                                                                        &             &                                 &                                 &                                 &                                 & \ding{109}                               & \ding{109}                               \\
  ESOC                                                                                                                        & $\square$   & $\square$                       &                                 &                                 & $\square$                       & \ding{109}$\square$                      & $\square$                                \\
  human-gene2                                                                                                                 &             &                                 &                                 &                                 & \ding{109}$\triangle$           & \ding{109}$\triangle$                    & \ding{109}$\triangle$                    \\
  IMDB                                                                                                                        &             &                                 &                                 & $\triangle$                     & $\triangle$                     & $\triangle$                              & $\triangle$                              \\
  kron-g500-logn16                                                                                                            &             & $\triangle$                     & $\triangle$                     & $\triangle$                     & $\triangle$                     & \ding{109}$\triangle$                    & \ding{109}$\triangle$                    \\
  Rucci1                                                                                                                      &             &                                 &                                 &                                 & $\square$                       &                                          &                                          \\
  sls                                                                                                                         & $\square$   & $\square$                       & $\square$                       & \ding{109}$\square$             & \ding{109}$\square$             & \ding{109}$\square$                      & \ding{109}$\square$                      \\
  Trec14                                                                                                                      &             &                                 &                                 &                                 &                                 &                                          & \ding{109}                               \\
   \bottomrule
  \end{longtable}
}
\begin{table}[h!]
\centering
  \caption*{
    \begin{tabular}{|l l|}
      \hline
      $\triangle$~:                                                                                                           & \KaHyPar{CA} exceeded time limit                                                                                                                                                                                                            \\
      \ding{108}~:                                                                                                            & \hMetis{R} exceeded time limit                                                                                                                                                                                                           \\
      \ding{109}~:                                                                                                            & \hMetis{K} exceeded time limit                                                                                                                                                                                                           \\
      $\square$~:                                                                                                             & \PaToH{Q} memory allocation error                                                                                                                                                                                                        \\
      \hline
    \end{tabular}
  }
\end{table}

\end{document}